\documentclass{llncs}
\makeatletter
\@ifundefined{lhs2tex.lhs2tex.sty.read}%
  {\@namedef{lhs2tex.lhs2tex.sty.read}{}%
   \newcommand\SkipToFmtEnd{}%
   \newcommand\EndFmtInput{}%
   \long\def\SkipToFmtEnd#1\EndFmtInput{}%
  }\SkipToFmtEnd

\newcommand\ReadOnlyOnce[1]{\@ifundefined{#1}{\@namedef{#1}{}}\SkipToFmtEnd}
\usepackage{amstext}
\usepackage{amssymb}
\usepackage{stmaryrd}
\DeclareFontFamily{OT1}{cmtex}{}
\DeclareFontShape{OT1}{cmtex}{m}{n}
  {<5><6><7><8>cmtex8
   <9>cmtex9
   <10><10.95><12><14.4><17.28><20.74><24.88>cmtex10}{}
\DeclareFontShape{OT1}{cmtex}{m}{it}
  {<-> ssub * cmtt/m/it}{}

\DeclareFontShape{OT1}{cmtt}{bx}{n}
  {<5><6><7><8>cmtt8
   <9>cmbtt9
   <10><10.95><12><14.4><17.28><20.74><24.88>cmbtt10}{}
\DeclareFontShape{OT1}{cmtex}{bx}{n}
  {<-> ssub * cmtt/bx/n}{}

\newcommand{\Varid}[1]{\mathit{#1}}
\newcommand{\anonymous}{\kern0.06em \vbox{\hrule\@width.5em}}

\usepackage{polytable}

\@ifundefined{mathindent}%
  {\newdimen\mathindent\mathindent\leftmargini}%
  {}%

\def\resethooks{%
  \global\let\SaveRestoreHook\empty
  \global\let\ColumnHook\empty}
\newcommand*{\savecolumns}[1][default]%
  {\g@addto@macro\SaveRestoreHook{\savecolumns[#1]}}
\newcommand*{\restorecolumns}[1][default]%
  {\g@addto@macro\SaveRestoreHook{\restorecolumns[#1]}}
\newcommand*{\aligncolumn}[2]%
  {\g@addto@macro\ColumnHook{\column{#1}{#2}}}

\resethooks

\newcommand{\onelinecommentchars}{\quad-{}- }
\newcommand{\commentbeginchars}{\enskip\{-}
\newcommand{\commentendchars}{-\}\enskip}

\newcommand{\visiblecomments}{%
  \let\onelinecomment=\onelinecommentchars
  \let\commentbegin=\commentbeginchars
  \let\commentend=\commentendchars}

\newcommand{\invisiblecomments}{%
  \let\onelinecomment=\empty
  \let\commentbegin=\empty
  \let\commentend=\empty}

\visiblecomments

\newlength{\blanklineskip}
\newcommand{\hsindent}[1]{\quad}
\let\hspre\empty
\let\hspost\empty

\EndFmtInput
\makeatother
\ReadOnlyOnce{polycode.fmt}%
\makeatletter

\newcommand{\hsnewpar}[1]%
  {{\parskip=0pt\parindent=0pt\par\vskip #1\noindent}}

\newcommand{\hscodestyle}{}

\newcommand{\sethscode}[1]%
  {\expandafter\let\expandafter\hscode\csname #1\endcsname
   \expandafter\let\expandafter\endhscode\csname end#1\endcsname}

  {\par\noindent
   \advance\leftskip\mathindent
   \hscodestyle
   \let\\=\@normalcr
   \let\hspre\(\let\hspost\)%
   \pboxed}%
  {\endpboxed\)%
   \par\noindent
   \ignorespacesafterend}

  {\hsnewpar\abovedisplayskip
   \advance\leftskip\mathindent
   \hscodestyle
   \let\hspre\(\let\hspost\)%
   \pboxed}%
  {\endpboxed%
   \hsnewpar\belowdisplayskip
   \ignorespacesafterend}

  {\hsnewpar\abovedisplayskip
   \advance\leftskip\mathindent
   \hscodestyle
   \let\\=\@normalcr
   \(\pboxed}%
  {\endpboxed\)%
   \hsnewpar\belowdisplayskip
   \ignorespacesafterend}

\newcommand{\plainhs}{\sethscode{plainhscode}}

\plainhs

  {\hsnewpar\abovedisplayskip
   \advance\leftskip\mathindent
   \hscodestyle
   \let\\=\@normalcr
   \(\parray}%
  {\endparray\)%
   \hsnewpar\belowdisplayskip
   \ignorespacesafterend}

  {\parray}{\endparray}

  {\(\parray}{\endparray\)}

\def\codeframewidth{\arrayrulewidth}
\RequirePackage{calc}

  {\parskip=\abovedisplayskip\par\noindent
   \hscodestyle
   \arrayrulewidth=\codeframewidth
   \tabular{@{}|p{\linewidth-2\arraycolsep-2\arrayrulewidth-2pt}|@{}}%
   \hline\framedhslinecorrect\\{-1.5ex}%
   \let\endoflinesave=\\
   \let\\=\@normalcr
   \(\pboxed}%
  {\endpboxed\)%
   \framedhslinecorrect\endoflinesave{.5ex}\hline
   \endtabular
   \parskip=\belowdisplayskip\par\noindent
   \ignorespacesafterend}

\newcommand{\framedhslinecorrect}[2]%
  {#1[#2]}

  {\(\def\column##1##2{}%
   \let\>\undefined\let\<\undefined\let\\\undefined
   \newcommand\>[1][]{}\newcommand\<[1][]{}\newcommand\\[1][]{}%
   \def\fromto##1##2##3{##3}%
   }{\) }%

  {\let\orighscode=\hscode
   \let\origendhscode=\endhscode
   \def\endhscode{\def\hscode{\endgroup\def\@currenvir{hscode}\\}\begingroup}
   \orighscode\def\hscode{\endgroup\def\@currenvir{hscode}}}%
  {\origendhscode
   \global\let\hscode=\orighscode
   \global\let\endhscode=\origendhscode}%

\makeatother
\EndFmtInput
\newif\ifextended
\newif\ifextendedtitle

\extendedtrue

\usepackage[small,bf]{caption}

\usepackage[svgnames]{xcolor}
\usepackage{hyperref}
\hypersetup{
  colorlinks=true,
  pdfborder={0 0 0},
  citecolor=DarkGreen,
  linkcolor=DarkBlue,
  urlcolor=DarkBlue,
  pdfauthor = {Stefan Heule, Deian Stefan, Edward Z. Yang, John C. Mitchell, Alejandro Russo},
  pdftitle = {IFC Inside: Retrofitting Languages with Dynamic Information Flow Control (Extended Version)}
}

\usepackage{array, amssymb, amsmath, mathpartir,thmtools}
\usepackage{balance}
\usepackage{tikz}
\newcommand{\Red}[1]{{\color{red} #1}}
\newcommand{\cut}[1]{}

\usepackage{wrapfig}
\usepackage[square,numbers,sort&compress]{natbib}


\usepackage{url}

\ifextended%
\else%
\let\llncssubparagraph\subparagraph
\let\subparagraph\paragraph
\usepackage[compact,nonindentfirst]{titlesec}
\titlespacing{\section}{0pt}{1.5ex}{0.4ex}
\titlespacing{\subsection}{0pt}{1ex}{0.37ex}
\titlespacing{\subsubsection}{0pt}{0.47ex}{0.0ex}
\let\subparagraph\llncssubparagraph
\fi%

\newcommand{\toref}[1]{}

\usepackage{float}
\newfloat{listing}{htbp}{lop} \floatname{listing}{Listing}

\usepackage{url}

\newcommand{\appendixextfirst}{%
  \ifextended%
  Appendix~\ref{sec:appendix-extensions}%
  \else%
  the appendix of the extended version of this paper%
  \fi%
}

\newcommand{\appref}[1]{%
  \ifextended%
  Appendix~\ref{#1}%
  \else%
  the extended version of this paper%
  \fi%
}

\newcommand{\codelink}{\url{http://github.com/deian/espectro}}
\newcommand{\nop}[1]{}

\begin{document}

\newcommand{\flows}{\sqsubseteq}
\newcommand{\lub}{\sqcup}
\newcommand{\glb}{\sqcap}

\newcommand{\ifc}[1]{\ensuremath{{\color{blue} #1}}}
\newcommand{\lcurr}{\ensuremath{\ifc{l_{\textrm{cur}}}}}




\newcommand{\tar}[1]{\ensuremath{\mathbf{\color{red} #1}}}



\newcommand{\Coloneqq}{::=} 

\newcommand{\dom}[1]{\ensuremath{{\textrm{dom}} #1}}
\newcommand{\fresh}[1]{\ensuremath{\textrm{fresh}(#1)}}






\title{
IFC Inside: Retrofitting Languages with Dynamic Information Flow Control
}
\ifextended
\subtitle{Extended Version}
\fi

\author{
 Stefan Heule\inst{1} \and
 Deian Stefan\inst{1} \and
 Edward Z. Yang\inst{1} \and
 John C. Mitchell\inst{1} \and
 Alejandro Russo\inst{2}\protect\footnote{Work partially done while at Stanford.}
}
\institute{Stanford University \and Chalmers University
}

\maketitle

\begin{abstract}
Many important security problems in JavaScript, such as browser
extension security, untrusted JavaScript libraries and safe
integration of mutually distrustful websites (mash-ups), may be
effectively addressed using an efficient implementation of information
flow control (IFC).  Unfortunately existing fine-grained approaches to
JavaScript IFC require modifications to the language semantics and its engine, a
non-goal for browser applications.  In this work, we take the ideas of
coarse-grained dynamic IFC and provide the theoretical
foundation for a language-based approach that can be applied to any
programming language for which external effects can be controlled.  We
then apply this formalism to server- and client-side JavaScript,
show how it generalizes to the C programming language, and connect it
to the Haskell LIO system.  Our methodology offers
design principles for the construction of information flow control
systems when isolation can easily be achieved, as well as
compositional proofs for optimized concrete implementations of these
systems, by relating them to their isolated variants.
\end{abstract}
\section{Introduction}
\label{sec:intro}

Modern web content is rendered using a potentially large number of
different components with differing provenance.
Disparate and untrusting components may arise from browser
extensions (whose JavaScript code runs alongside website
code), web applications (with possibly untrusted third-party
libraries), and mashups (which combine code and data from
websites that may not even be aware of each other's existence.)
While just-in-time combination of untrusting components
offers great flexibility, it also poses complex security challenges.
In particular, maintaining data privacy in the face of malicious
extensions, libraries, and mashup components has been difficult.

Information flow control (IFC) is a promising technique
that provides security
by tracking the flow of sensitive data through a system.
Untrusted code is confined so that it cannot exfiltrate data, except as
per an information flow policy.  Significant research has been devoted to
adding various forms of IFC to different kinds of programming languages
and systems.  In the context of the web, however, there is a strong
motivation to preserve JavaScript's semantics and avoid
JavaScript-engine modifications, while retrofitting it with dynamic information
flow control.

The Operating Systems community has tackled this challenge (e.g.,
in~\cite{Zeldovich:2006}) by taking a \textit{coarse-grained} approach
to IFC: dividing an application into coarse computational units,
each with a single label dictating its security policy, and only
monitoring communication between them.
This coarse-grained approach provides a number of advantages when
compared to the fine-grained approaches typically employed by language-based systems.
First, adding IFC does not require intrusive changes to an
existing programming language, thereby also allowing
the reuse of existing programs.  Second, it has a small
runtime overhead because checks need only
be performed at isolation boundaries
instead of (almost) every program instruction~(e.g.,~\cite{JSFlow}).
Finally, associating
a single security label with the entire computational unit simplifies
understanding and reasoning about the security guarantees of the
system, without reasoning about most of the
technical details of the semantics of the underlying programming language.

In this paper, we present a framework which brings coarse-grained IFC
ideas into a language-based setting:
an information flow control system should
be thought of as multiple instances of completely isolated language
runtimes or \emph{tasks}, with information flow control applied to
inter-task communication.  We describe a formal system in which an IFC
system can be designed once and then applied to any programming language
which has control over external effects (e.g., JavaScript or C with
access to hardware privilege separation).  We formalize this system
using an approach by
Matthews and Findler~\cite{Matthews:2007:OSM:1190216.1190220} for combining
operational semantics and prove non-interference guarantees
that are independent of the choice of a specific target language.

There are a number of points that distinguish this setting from
previous coarse-grained IFC systems.
First, even though the underlying semantic model involves communicating
tasks, these tasks can be coordinated together in ways that simulate
features of traditional languages.
In fact, simulating
features in this way is a useful \emph{design tool} for discovering
what variants of the features are permissible and which are not.
Second, although completely separate tasks are semantically easy to
reason about, real-world implementations often blur the lines between
tasks in the name of efficiency.
Characterizing what optimizations are permissible is subtle, since
removing transitions from the operational semantics of a language can
break non-interference.  We partially address this issue
by characterizing isomorphisms between the operational semantics of our
abstract language and a concrete implementation, showing that if this
relationship holds, then non-interference in the abstract specification
carries over to the concrete implementation.

Our contributions can be summarized as follows:
\vspace*{-0.3em}
\begin{itemize}
  \item We give formal semantics for a core coarse-grained
  dynamic information flow control language free of non-IFC constructs.
  We then show how a large class of target languages can be combined
  with this IFC language and prove that the result provides
  non-interference. (Sections~\ref{sec:retrofit} and \ref{sec:formal})
  \item We provide a proof technique to show the non-interference
  of a concrete semantics for a potentially optimized IFC language
  by means of an isomorphism and show a class of restrictions on
  the IFC language that preserves non-interference. (Section~\ref{sec:concrete})
  \item We have implemented an IFC system based on these semantics
  for Node.js, and we connect our formalism to another implementation
  based on this work for client-side JavaScript~\cite{swapi}.
  Furthermore, we outline an implementation for the C programming
  language and describe improvements to the Haskell LIO system that
  resulted from this framework.
  (Section~\ref{sec:real})
\end{itemize}

\ifextended
\else
In the extended version of this paper we give all the relevant proofs and
extend our IFC language with additional features~\cite{extended}.
\fi

\section{Retrofitting Languages with IFC}
\label{sec:retrofit}

Before moving on to the formal treatment of our system, we give
a brief primer of information flow control and describe some example programs
in our system, emphasizing the parallel between their implementation
in a multi-task setting, and the traditional, ``monolithic'' programming
language feature they simulate.

Information flow control systems operate by associating data with \emph{labels},
and specifying whether or not data tagged with one label \ensuremath{l_{1}} can flow
to another
label \ensuremath{l_{2}} (written as \ensuremath{l_{1}\flows{}l_{2}}).  These labels encode the desired
security policy (for example, confidential information should not flow to
a public channel), while the work of specifying the semantics of an information
flow language involves demonstrating that impermissible flows cannot happen,
a property called \emph{non-interference}~\cite{Goguen82}.
In our coarse-grained floating-label approach, labels are associated with tasks.
The task label---we refer to the label of the currently executing task as the
\emph{current label}---serves to protect everything in the task's scope;
all data in a task shares this common label.

As an example, here is a program which spawns a new isolated task,
and then sends it a mutable reference:
\begin{align*}
    & \ensuremath{\mathbf{let}\;\Varid{i}\mathrel{=}\tar{_{\textrm{TI}}\lfloor}\mathbf{sandbox}\;(\mathbf{blockingRecv}\ \Varid{x},\anonymous \;\mathbf{in}\;\ifc{^{\textrm{IT}}\lceil}\mathbin{!}\tar{_{\textrm{TI}}\lfloor}\Varid{x}\tar{\rfloor}\ifc{\rceil})\tar{\rfloor}}\\
    & \ensuremath{\mathbf{in}\;\tar{_{\textrm{TI}}\lfloor}\mathbf{send}\;\ifc{^{\textrm{IT}}\lceil}\Varid{i}\ifc{\rceil}\;l\;\ifc{^{\textrm{IT}}\lceil}\mathbf{ref}\;\mathbf{true}\ifc{\rceil}\tar{\rfloor}}
\end{align*}
For now, ignore the tags \ensuremath{\tar{_{\textrm{TI}}\lfloor}\cdot \tar{\rfloor}} and \ensuremath{\ifc{^{\textrm{IT}}\lceil}\cdot \ifc{\rceil}}: roughly, this code creates a new
\ensuremath{\mathbf{sandbox}}ed task with identifier $i$ which waits (\textbf{blockingRecv}, binding $x$ with the received message) for a
message, and then \textbf{send}s the task a mutable reference \ensuremath{(\mathbf{ref}\;\mathbf{true})} which it labels $l$.  If this operation actually shared the mutable cell between the two tasks, it
could be used to violate information flow control if the tasks had
differing labels.  At this point, the designer of an IFC system might
add label checks to mutable references, to check the labels of the
reader and writer. While this solves the leak, for languages like
JavaScript, where references are prevalently used, this also dooms the
performance of the system.

Our design principles suggest a different resolution: when these
constructs are treated as isolated tasks, each of which have their own heaps, it
is obviously the case that there is no sharing; in fact, the sandboxed task receives a dangling pointer.  Even if there is only one heap, if we enforce that references
not be shared, the two systems are morally equivalent. (We elaborate on
this formally in Section~\ref{sec:concrete}.)  Finally, this
semantics strongly suggests that one should restrict the types of
data which may be passed between tasks (for example, in JavaScript, one
might only allow JSON objects to be passed between tasks, rather than
general object structures).

Existing language-based, coarse-grained IFC
systems~\cite{Hritcu:2013:YIB:2497621.2498098,stefan:2012:arxiv-flexible}
allow a sub-computation to temporarily raise the floating-label; after
the sub-computation is done, the floating-label is restored to its
original label. When this occurs, the enforcement mechanism must
ensure that information does not leak to the (less confidential)
program continuation. The presence of exceptions adds yet more
intricacies.  For instance, exceptions should not automatically
propagate from a sub-computation directly into the program
continuation, and, if such exceptions are allowed to be inspected, the
floating-label at the point of the exception-raise must be tracked
alongside the exception
value~\cite{Hritcu:2013:YIB:2497621.2498098,stefan:2012:arxiv-flexible,Hedin:2012}.
In contrast, our system
provides the same flexibility and guarantees with no extra checks: tasks
are used to execute sub-computations, but the mere definition of
isolated tasks guarantees that (a) tasks only transfer data to the
program continuation by using inter-task communication means, and (b)
exceptions do cross tasks boundaries automatically.

\subsection{Preliminaries}

Our goal now is to
describe how to take a \textbf{{\color{red} target
language}} with a formal operational semantics and combine it with an
\textit{{\color{blue} information flow control language}}.  For example,
taking ECMAScript as the target language and combining it with our IFC
language should produce the formal semantics for the core part of COWL~\cite{swapi}.  In this
presentation, we use a simple, untyped lambda calculus with mutable
references and fixpoint in place of ECMAScript to demonstrate some the key
properties of the system (and, because the embedding does not care
about the target language features); we discuss the proper embedding
in more detail in Section~\ref{sec:real}.

\vspace{1pt}
\noindent
\textit{Notation}
We have typeset nonterminals of the target language using \textbf{{\color{red}
bold font}} while the nonterminals of the IFC language have been typeset
with \textit{{\color{blue} italic font}}.  Readers are encouraged to view
a color copy of this paper, where target language nonterminals are colored \textbf{{\color{red} red}}
and IFC language nonterminals are colored \textit{{\color{blue} blue}}.

\subsection{Target Language: Mini-ES}

In Fig.~\ref{fig:ml}, we give a simple, untyped lambda calculus with
mutable references and fixpoint, prepared for combination with an
information flow control language.  The presentation is mostly standard, and utilizes Felleisen-Hieb reduction
semantics~\cite{Felleisen:1992:RRS:136293.136297} to define the
operational semantics of the system.  One peculiarity is that our language
defines an evaluation context \ensuremath{\tar{E}}, but, the evaluation rules have been
expressed in terms of a different evaluation context \ensuremath{\mathcal{E}_{\tar{\Sigma}}};
Here, we follow the approach of Matthews and
Findler~\cite{Matthews:2007:OSM:1190216.1190220} in order to simplify combining
semantics of multiple languages.
To derive the usual operational semantics for this language, the evaluation
context merely needs to be defined as \ensuremath{\mathcal{E}_{\tar{\Sigma}}\left[\tar{e}\right]\triangleq{}\tar{\Sigma},\tar{E}\left[\tar{e}\right]}.
However, when we combine this language with an IFC language, we
reinterpret the meaning of this evaluation context.

In general, we require that a target language be expressed in terms
of some global machine state \ensuremath{\tar{\Sigma}}, some evaluation context \ensuremath{\tar{E}},
some expressions \ensuremath{\tar{e}}, some set of values \ensuremath{\tar{v}} and a \emph{deterministic}
reduction relation on full configurations $\ensuremath{\tar{\Sigma}} \times \ensuremath{\tar{E}} \times \ensuremath{\tar{e}}$.

\begin{figure}[t]
\begin{hscode}\SaveRestoreHook
\column{B}{@{}>{\hspre}l<{\hspost}@{}}%
\column{6}{@{}>{\hspre}l<{\hspost}@{}}%
\column{22}{@{}>{\hspre}l<{\hspost}@{}}%
\column{40}{@{}>{\hspre}l<{\hspost}@{}}%
\column{47}{@{}>{\hspre}l<{\hspost}@{}}%
\column{E}{@{}>{\hspre}l<{\hspost}@{}}%
\>[B]{}\tar{v}{}\<[6]%
\>[6]{}\Coloneqq\lambda \tar{x}.\tar{e}\mid \mathbf{true}\mid \mathbf{false}\mid \tar{a}{}\<[E]%
\\
\>[B]{}\tar{e}{}\<[6]%
\>[6]{}\Coloneqq\tar{v}\mid \tar{x}\mid \tar{e}\;\tar{e}\mid \mathbf{if}\;\tar{e}\;\mathbf{then}\;\tar{e}\;\mathbf{else}\;\tar{e}\mid \mathbf{ref}\;\tar{e}\mid \mathbin{!}\tar{e}\mid \tar{e}\mathbin{:=}\tar{e}\mid \mathbf{fix}\;\tar{e}{}\<[E]%
\\
\>[B]{}\tar{E}{}\<[6]%
\>[6]{}\Coloneqq\tar{[\cdot ]_T}\mid \tar{E}\;\tar{e}\mid \tar{v}\;\tar{E}\mid \mathbf{if}\;\tar{E}\;\mathbf{then}\;\tar{e}\;\mathbf{else}\;\tar{e}\mid \mathbf{ref}\;\tar{E}\mid \mathbin{!}\tar{E}\mid \tar{E}\mathbin{:=}\tar{e}\mid \tar{v}\mathbin{:=}\tar{E}\mid \mathbf{fix}\;\tar{E}{}\<[E]%
\\
\>[B]{}\tar{e}_{1};\tar{e}_{2}{}\<[22]%
\>[22]{}\triangleq{}(\lambda \tar{x}.\tar{e}_{2})\;\tar{e}_{1}\;{}\<[40]%
\>[40]{}\mathbf{where}\;{}\<[47]%
\>[47]{}\tar{x}\;\not\in\;\mathcal{FV}\;(\tar{e}_{2}){}\<[E]%
\\
\>[B]{}\mathbf{let}\;\tar{x}\mathrel{=}\tar{e}_{1}\;\mathbf{in}\;\tar{e}_{2}{}\<[22]%
\>[22]{}\triangleq{}(\lambda \tar{x}.\tar{e}_{2})\;\tar{e}_{1}{}\<[E]%
\ColumnHook
\end{hscode}\resethooks
\begin{mathpar}

\inferrule[T-app]
{ } {\ensuremath{\mathcal{E}_{\tar{\Sigma}}\left[(\lambda \Varid{x}.\tar{e})\;\tar{v}\right]\rightarrow\mathcal{E}_{\tar{\Sigma}}\left[\{\mskip1.5mu \tar{v}\mathbin{/}\Varid{x}\mskip1.5mu\}\;\tar{e}\right]}}

\and
\inferrule[T-ifTrue]
{ } {\ensuremath{\mathcal{E}_{\tar{\Sigma}}\left[\;\mathbf{if}\;\mathbf{true}\;\mathbf{then}\;\tar{e}_{1}\;\mathbf{else}\;\tar{e}_{2}\right]\rightarrow\mathcal{E}_{\tar{\Sigma}}\left[\tar{e}_{1}\right]}}

\end{mathpar}

\caption{\ensuremath{\Red{\lambda_{\text{ES}}}}: simple untyped lambda calculus extended with booleans,
mutable references and general recursion.  For space reasons we only show two
representative reduction rules;  full rules can be found in \appref{sec:app:semantics}.}
\label{fig:ml}
\end{figure}

\subsection{IFC Language}

As mentioned previously, most modern, dynamic information flow control
languages encode policy by associating a label with data.  Our
embedding is agnostic to the choice of labeling scheme; we only require
the labels to form a lattice~\cite{Denning:1976:LMS:360051.360056} with
the partial order $\sqsubseteq$, join \ensuremath{\lub}, and meet \ensuremath{\glb}.  In this
paper, we simply represent labels with the metavariable $l$, but do not
discuss them in more detail.
To enforce labels, the IFC monitor inspects the
current label  before performing a read or a write to decide whether the operation is permitted.
A task can only write to entities that are at least as sensitive.
Similarly, it can only read from entities that are less sensitive.
However, as in other floating-label systems, this current label can be raised
to allow the task to read from more sensitive entities at the cost of giving up
the ability to write to others. 

In Fig.~\ref{fig:ifc}, we give the syntax and \emph{single-task}
evaluation rules for a minimal information flow control language.
Ordinarily, information flow control languages are defined by directly
stating a base language plus information flow control operators.  In
contrast, our language is purposely minimal: it does not have sequencing
operations, control flow, or other constructs.  However, it contains
support for the following core information flow control features:

\begin{itemize}
    \item First-class labels, with label values $l$ as well as operations for computing on
labels (\ensuremath{\flows{}}, \ensuremath{\lub} and \ensuremath{\glb}).
    \item Operations for inspecting (\textbf{getLabel}) and modifying
    (\textbf{setLabel}) the current label of the task (a task can only increase its label).
    \item Operations for non-blocking inter-task communication (\textbf{send}
    and \textbf{recv}), which interact with the global store of per-task
    message queues \ensuremath{\ifc{\Sigma}}.
    \item A sandboxing operation used to spawn new isolated tasks. In
    concurrent settings \ensuremath{\mathbf{sandbox}} corresponds to a fork-like primitive,
    whereas in a
    sequential setting, it more closely resembles
    computations which might temporarely raise the current
    floating-label~\cite{lio,Hritcu:2013:YIB:2497621.2498098}.

\end{itemize}

These operations are all defined with respect to an evaluation context
\ensuremath{\mathcal{E}_{\ifc{\Sigma}}^{\ifc{\Varid{i}},\ifc{l}}} that represents the context of the current task.
The evaluation context has three important pieces of
state: the global message queues \ensuremath{\ifc{\Sigma}}, the current label \ensuremath{\ifc{l}} and the task ID \ensuremath{\ifc{\Varid{i}}}.

We note that first-class labels, tasks (albeit named differently), and
operations for inspecting the current label are essentially universal to
all floating-label systems.
However, our choice of communication primitives is motivated by
those present in browsers, namely \texttt{postMessage}~\cite{webmessaging}.
Of course, other choices, such as blocking communication or labeled channels,
are possible. 

These asynchronous communication primitives are worth further
discussion.  When a task is sending a message using \ensuremath{\mathbf{send}}, it also labels that
message with a label \ensuremath{\ifc{l}'} (which must be at or above the task's current label \ensuremath{\ifc{l}}).
Messages can only be received by a task if its current label is
at least as high as the label of the message.
Specifically, receiving a message using
$\ensuremath{\mathbf{recv}\ \ifc{\Varid{x}}_{1},\ifc{\Varid{x}}_{2}\;\mathbf{in}\;\ifc{\Varid{e}}_{1}\ \mathbf{else}\ \ifc{\Varid{e}}_{2}}$
binds the message and the sender's task identifier
to local variables \ensuremath{\ifc{\Varid{x}}_{1}} and \ensuremath{\ifc{\Varid{x}}_{2}}, respectively, and then executes \ensuremath{\ifc{\Varid{e}}_{1}}.
Otherwise, if there are no messages, that task continues its execution with \ensuremath{\ifc{\Varid{e}}_{2}}.
We denote the filtering of the message queue by \ensuremath{\ifc{\Theta}\preceq\ifc{l}},
which is defined as follows.
If \ensuremath{\ifc{\Theta}} is the empty list \ensuremath{\mathbf{nil}}, the
function is simply the identity function, i.e.,
\ensuremath{\mathbf{nil}\preceq\ifc{l}\mathrel{=}\mathbf{nil}}, and otherwise:
\[
\ensuremath{((\ifc{l}',\ifc{\Varid{i}},\ifc{\Varid{e}}),\ifc{\Theta})\preceq\ifc{l}} = \left\{
\begin{array}{l l}
\ensuremath{(\ifc{l}',\ifc{\Varid{i}},\ifc{\Varid{e}}),(\ifc{\Theta}\preceq\ifc{l})} & \quad \text{if \ensuremath{\ifc{l}'\;\flows{}\;\ifc{l}}}\\
\ensuremath{\ifc{\Theta}\preceq\ifc{l}} & \quad \text{otherwise}
\end{array} \right.
\]
This ensures that tasks cannot receive messages that are more sensitive
than their current label would allow.

\begin{figure}
\begin{hscode}\SaveRestoreHook
\column{B}{@{}>{\hspre}l<{\hspost}@{}}%
\column{6}{@{}>{\hspre}c<{\hspost}@{}}%
\column{6E}{@{}l@{}}%
\column{8}{@{}>{\hspre}l<{\hspost}@{}}%
\column{11}{@{}>{\hspre}l<{\hspost}@{}}%
\column{24}{@{}>{\hspre}l<{\hspost}@{}}%
\column{38}{@{}>{\hspre}l<{\hspost}@{}}%
\column{43}{@{}>{\hspre}c<{\hspost}@{}}%
\column{43E}{@{}l@{}}%
\column{48}{@{}>{\hspre}l<{\hspost}@{}}%
\column{58}{@{}>{\hspre}l<{\hspost}@{}}%
\column{63}{@{}>{\hspre}l<{\hspost}@{}}%
\column{71}{@{}>{\hspre}l<{\hspost}@{}}%
\column{76}{@{}>{\hspre}c<{\hspost}@{}}%
\column{76E}{@{}l@{}}%
\column{81}{@{}>{\hspre}l<{\hspost}@{}}%
\column{E}{@{}>{\hspre}l<{\hspost}@{}}%
\>[B]{}\ifc{\Varid{v}}{}\<[6]%
\>[6]{}\Coloneqq{}\<[6E]%
\>[11]{}\ifc{\Varid{i}}\mid \ifc{l}\mid \mathbf{true}\mid \mathbf{false}\mid \langle\rangle\;\qquad{}\;\qquad{}\;\otimes{}\<[58]%
\>[58]{}\Coloneqq{}\<[63]%
\>[63]{}\flows{}\mid \lub\mid \glb{}\<[E]%
\\
\>[B]{}\ifc{\Varid{e}}{}\<[6]%
\>[6]{}\Coloneqq{}\<[6E]%
\>[11]{}\ifc{\Varid{v}}\mid \ifc{\Varid{x}}\mid \ifc{\Varid{e}}\;\otimes\;\ifc{\Varid{e}}\mid \mathbf{getLabel}\mid \mathbf{setLabel}\;\ifc{\Varid{e}}\mid \mathbf{taskId}\mid {}\<[71]%
\>[71]{}\mathbf{sandbox}\;\ifc{\Varid{e}}{}\<[E]%
\\
\>[6]{}\hsindent{2}{}\<[8]%
\>[8]{}\mid \mathbf{send}\;\ifc{\Varid{e}}\;\ifc{\Varid{e}}\;\ifc{\Varid{e}}\mid \mathbf{recv}\ \ifc{\Varid{x}},\ifc{\Varid{x}}\;\mathbf{in}\;\ifc{\Varid{e}}\ \mathbf{else}\ \ifc{\Varid{e}}{}\<[E]%
\\
\>[B]{}\ifc{E}{}\<[6]%
\>[6]{}\Coloneqq{}\<[6E]%
\>[11]{}\ifc{[\cdot ]_I}\mid \ifc{E}\;\otimes\;\ifc{\Varid{e}}\mid \ifc{\Varid{v}}\;\otimes\;\ifc{E}\mid \mathbf{setLabel}\;\ifc{E}\mid {}\<[58]%
\>[58]{}\mathbf{send}\;\ifc{E}\;\ifc{\Varid{e}}\;\ifc{\Varid{e}}\mid \mathbf{send}\;\ifc{\Varid{v}}\;\ifc{E}\;\ifc{\Varid{e}}\mid \mathbf{send}\;\ifc{\Varid{v}}\;\ifc{\Varid{v}}\;\ifc{E}{}\<[E]%
\\
\>[B]{}\ifc{\theta}{}\<[6]%
\>[6]{}\Coloneqq{}\<[6E]%
\>[11]{}(\ifc{l},\ifc{\Varid{i}}\;\ifc{\Varid{e}})\;{}\<[24]%
\>[24]{}\qquad{}\;\qquad{}\;{}\<[38]%
\>[38]{}\ifc{\Theta}{}\<[43]%
\>[43]{}\Coloneqq{}\<[43E]%
\>[48]{}\mathbf{nil}\mid \ifc{\theta},\ifc{\Theta}\;{}\<[63]%
\>[63]{}\qquad{}\;{}\<[71]%
\>[71]{}\ifc{\Sigma}{}\<[76]%
\>[76]{}\Coloneqq{}\<[76E]%
\>[81]{}\emptyset\mid \ifc{\Sigma}\left[\ifc{\Varid{i}}\mapsto{}\ifc{\Theta}\right]{}\<[E]%
\ColumnHook
\end{hscode}\resethooks
\begin{mathpar}

\inferrule[I-getTaskId]
{ }
{\ensuremath{\mathcal{E}_{\ifc{\Sigma}}^{\ifc{\Varid{i}},\ifc{l}}\left[\mathbf{taskId}\right]\rightarrow\mathcal{E}_{\ifc{\Sigma}}^{\ifc{\Varid{i}},\ifc{l}}\left[\ifc{\Varid{i}}\right]}}
\hspace*{-0.5cm}

\and
\inferrule[I-getLabel]
{ }
{\ensuremath{\mathcal{E}_{\ifc{\Sigma}}^{\ifc{\Varid{i}},\ifc{l}}\left[\mathbf{getLabel}\right]\rightarrow\mathcal{E}_{\ifc{\Sigma}}^{\ifc{\Varid{i}},\ifc{l}}\left[\ifc{l}\right]}}
\hspace*{-0.5cm}

\and
\inferrule[I-labelOp]
{ \ensuremath{\llbracket \ifc{l}_{1}\;\otimes\;\ifc{l}_{2}\rrbracket\mathrel{=}\ifc{\Varid{v}}}}
{ \ensuremath{\mathcal{E}_{\ifc{\Sigma}}^{\ifc{\Varid{i}},\ifc{l}}\left[\ifc{l}_{1}\;\otimes\;\ifc{l}_{2}\right]\rightarrow\mathcal{E}_{\ifc{\Sigma}}^{\ifc{\Varid{i}},\ifc{l}}\left[\ifc{\Varid{v}}\right]} }

\and
\inferrule[I-send]
{
\ensuremath{\ifc{l}\;\flows{}\;\ifc{l}'}\\
\ensuremath{\ifc{\Sigma}(\ifc{\Varid{i}}')\mathrel{=}\ifc{\Theta}}\\
\ensuremath{\ifc{\Sigma}'\mathrel{=}\ifc{\Sigma}\left[\ifc{\Varid{i}}'\mapsto{}(\ifc{l}',\ifc{\Varid{i}},\ifc{\Varid{v}}),\ifc{\Theta}\right]}
}
{\ensuremath{\mathcal{E}_{\ifc{\Sigma}}^{\ifc{\Varid{i}},\ifc{l}}\left[\mathbf{send}\;\ifc{\Varid{i}}'\;\ifc{l}'\;\ifc{\Varid{v}}\right]\rightarrow\mathcal{E}_{\ifc{\Sigma}'}^{\ifc{\Varid{i}},\ifc{l}}\left[\langle\rangle\right]}}

\and
\inferrule[I-recv]
{
\ensuremath{(\ifc{\Sigma}(\ifc{\Varid{i}})\preceq\ifc{l})\mathrel{=}\ifc{\theta}_{1},\mathbin{...},\ifc{\theta}_\Varid{k},(\ifc{l}',\ifc{\Varid{i}}',\ifc{\Varid{v}})}\\
\ensuremath{\ifc{\Sigma}'\mathrel{=}\ifc{\Sigma}\left[\ifc{\Varid{i}}\mapsto{}(\ifc{\theta}_{1},\mathbin{...},\ifc{\theta}_\Varid{k})\right]}\\
}
{\ensuremath{\mathcal{E}_{\ifc{\Sigma}}^{\ifc{\Varid{i}},\ifc{l}}\left[\mathbf{recv}\ \ifc{\Varid{x}}_{1},\ifc{\Varid{x}}_{2}\;\mathbf{in}\;\ifc{\Varid{e}}_{1}\ \mathbf{else}\ \ifc{\Varid{e}}_{2}\right]\rightarrow\mathcal{E}_{\ifc{\Sigma}'}^{\ifc{\Varid{i}},\ifc{l}}\left[\{\mskip1.5mu \ifc{\Varid{v}}\mathbin{/}\ifc{\Varid{x}}_{1},\ifc{\Varid{i}}'\mathbin{/}\ifc{\Varid{x}}_{2}\mskip1.5mu\}\;\ifc{\Varid{e}}_{1}\right]}}

\and
\inferrule[I-noRecv]
{
\ensuremath{\ifc{\Sigma}(\ifc{\Varid{i}})\preceq\ifc{l}\mathrel{=}\mathbf{nil}}\\
\ensuremath{\ifc{\Sigma}'\mathrel{=}\ifc{\Sigma}\left[\ifc{\Varid{i}}\mapsto{}\mathbf{nil}\right]}
}
{\ensuremath{\mathcal{E}_{\ifc{\Sigma}}^{\ifc{\Varid{i}},\ifc{l}}\left[\mathbf{recv}\ \ifc{\Varid{x}}_{1},\ifc{\Varid{x}}_{2}\;\mathbf{in}\;\ifc{\Varid{e}}_{1}\ \mathbf{else}\ \ifc{\Varid{e}}_{2}\right]\rightarrow\mathcal{E}_{\ifc{\Sigma}'}^{\ifc{\Varid{i}},\ifc{l}}\left[\ifc{\Varid{e}}_{2}\right]}}

\and
\inferrule[I-setLabel]
{ \ensuremath{\ifc{l}\;\flows{}\;\ifc{l}'} }
{\ensuremath{\mathcal{E}_{\ifc{\Sigma}}^{\ifc{\Varid{i}},\ifc{l}}\left[\mathbf{setLabel}\;\ifc{l}'\right]\rightarrow\mathcal{E}_{\ifc{\Sigma}}^{\ifc{\Varid{i}},\ifc{l}'}\left[\langle\rangle\right]}}

\end{mathpar}
\caption{IFC language with all single-task operations.}
\label{fig:ifc}
\end{figure}

\subsection{The Embedding}
Fig.~\ref{fig:embedding}
provides all of the rules responsible for actually carrying out the embedding of the IFC language within the target language.
The most important feature of this embedding is that every task maintains its own
copy of the target language global state and evaluation context, thus
enforcing isolation between various tasks.  In more detail:

\begin{figure}
\begin{tabular}{lll}
\begin{minipage}{.22\textwidth}
\begin{hscode}\SaveRestoreHook
\column{B}{@{}>{\hspre}l<{\hspost}@{}}%
\column{5}{@{}>{\hspre}l<{\hspost}@{}}%
\column{E}{@{}>{\hspre}l<{\hspost}@{}}%
\>[B]{}\ifc{\Varid{v}}{}\<[5]%
\>[5]{}\Coloneqq\cdots\mid \ifc{^{\textrm{IT}}\lceil}\tar{v}\ifc{\rceil}{}\<[E]%
\\
\>[B]{}\ifc{\Varid{e}}{}\<[5]%
\>[5]{}\Coloneqq\cdots\mid \ifc{^{\textrm{IT}}\lceil}\tar{e}\ifc{\rceil}{}\<[E]%
\\
\>[B]{}\ifc{E}{}\<[5]%
\>[5]{}\Coloneqq\cdots\mid \ifc{^{\textrm{IT}}\lceil}\tar{E}\ifc{\rceil}{}\<[E]%
\ColumnHook
\end{hscode}\resethooks
\end{minipage} &
\begin{minipage}{.22\textwidth}
\begin{hscode}\SaveRestoreHook
\column{B}{@{}>{\hspre}l<{\hspost}@{}}%
\column{5}{@{}>{\hspre}l<{\hspost}@{}}%
\column{E}{@{}>{\hspre}l<{\hspost}@{}}%
\>[B]{}\tar{v}{}\<[5]%
\>[5]{}\Coloneqq\cdots\mid \tar{_{\textrm{TI}}\lfloor}\ifc{\Varid{v}}\tar{\rfloor}{}\<[E]%
\\
\>[B]{}\tar{e}{}\<[5]%
\>[5]{}\Coloneqq\cdots\mid \tar{_{\textrm{TI}}\lfloor}\ifc{\Varid{e}}\tar{\rfloor}{}\<[E]%
\\
\>[B]{}\tar{E}{}\<[5]%
\>[5]{}\Coloneqq\cdots\mid \tar{_{\textrm{TI}}\lfloor}\ifc{E}\tar{\rfloor}{}\<[E]%
\ColumnHook
\end{hscode}\resethooks
\end{minipage} &
\begin{minipage}{.42\textwidth}
\begin{hscode}\SaveRestoreHook
\column{B}{@{}>{\hspre}l<{\hspost}@{}}%
\column{E}{@{}>{\hspre}l<{\hspost}@{}}%
\>[B]{}\mathcal{E}_{\tar{\Sigma}}\left[\tar{e}\right]\triangleq{}\ifc{\Sigma};\langle \tar{\Sigma}, \ifc{E}\ifc{[}\tar{e}\ifc{]}_\tar{T}\rangle^{\ifc{\Varid{i}}}_{\ifc{l}},\ldots{}\<[E]%
\\
\>[B]{}\mathcal{E}_{\ifc{\Sigma}}^{\ifc{\Varid{i}},\ifc{l}}\left[\ifc{\Varid{e}}\right]\triangleq{}\ifc{\Sigma};\langle \tar{\Sigma}, \ifc{E}\ifc{[}\ifc{\Varid{e}}\ifc{]}_\ifc{I}\rangle^{\ifc{\Varid{i}}}_{\ifc{l}},\ldots{}\<[E]%
\\
\>[B]{}\mathcal{E}\;[\mskip1.5mu \Varid{e}\mskip1.5mu]\rightarrow\ifc{\Sigma};\ifc{\Varid{t}},\ldots\triangleq{}\mathcal{E}\;[\mskip1.5mu \Varid{e}\mskip1.5mu]\overset{\alpha}{\hookrightarrow}\ifc{\Sigma};\alpha_{{\tiny\mathrm{\text{step}}}}(\ifc{\Varid{t}},\ldots){}\<[E]%
\ColumnHook
\end{hscode}\resethooks
\end{minipage}
\end{tabular}
\begin{mathpar}
\inferrule[I-sandbox]
{
\ensuremath{\ifc{\Sigma}'\mathrel{=}\ifc{\Sigma}\left[\ifc{\Varid{i}}'\mapsto{}\mathbf{nil}\right]}\\
\ensuremath{\tar{\Sigma}'\mathrel{=}\kappa\;(\tar{\Sigma})}\\
\ensuremath{\ifc{\Varid{t}}_{1}\mathrel{=}\langle \tar{\Sigma}, \ifc{E}\ifc{[}\ifc{\Varid{i}}'\ifc{]}\rangle^{\ifc{\Varid{i}}}_{\ifc{l}}}\\
\ensuremath{\ifc{\Varid{t}}_\textrm{new}\mathrel{=}\langle \tar{\Sigma}', \ifc{\Varid{e}}\rangle^{\ifc{\Varid{i}}'}_{\ifc{l}}}\\
\ensuremath{\textrm{fresh}(\ifc{\Varid{i}}')}
}
{\ensuremath{\ifc{\Sigma};\langle \tar{\Sigma}, \ifc{E}\ifc{[}\mathbf{sandbox}\;\ifc{\Varid{e}}\ifc{]}_\ifc{I}\rangle^{\ifc{\Varid{i}}}_{\ifc{l}},\ldots\overset{\alpha}{\hookrightarrow}\ifc{\Sigma}';\alpha_{{\tiny\mathrm{\text{sandbox}}}}(\ifc{\Varid{t}}_{1},\ldots,\ifc{\Varid{t}}_\textrm{new})}}

\and
\inferrule[I-done]
{\ensuremath{}}
{\ensuremath{\ifc{\Sigma};\langle \tar{\Sigma}, \ifc{\Varid{v}}\rangle^{\ifc{\Varid{i}}}_{\ifc{l}},\ldots\overset{\alpha}{\hookrightarrow}\ifc{\Sigma};\alpha_{{\tiny\mathrm{\text{done}}}}(\langle \tar{\Sigma}, \ifc{\Varid{v}}\rangle^{\ifc{\Varid{i}}}_{\ifc{l}},\ldots)}}

\and
\inferrule[I-noStep]
{\ensuremath{\ifc{\Sigma};\ifc{\Varid{t}},\ldots\not\overset{\alpha}{\hookrightarrow}}}
{\ensuremath{\ifc{\Sigma};\ifc{\Varid{t}},\ldots\overset{\alpha}{\hookrightarrow}\ifc{\Sigma};\alpha_{{\tiny\mathrm{\text{noStep}}}}(\ifc{\Varid{t}},\ldots)}}

\and
\inferrule[I-border]
{ }
{\ensuremath{\mathcal{E}_{\ifc{\Sigma}}^{\ifc{\Varid{i}},\ifc{l}}\left[\ifc{^{\textrm{IT}}\lceil}\tar{_{\textrm{TI}}\lfloor}\ifc{\Varid{e}}\tar{\rfloor}\ifc{\rceil}\right]\rightarrow\mathcal{E}_{\ifc{\Sigma}}^{\ifc{\Varid{i}},\ifc{l}}\left[\ifc{\Varid{e}}\right]}}

\and
\inferrule[T-border]
{ } {\ensuremath{\mathcal{E}_{\tar{\Sigma}}\left[\tar{_{\textrm{TI}}\lfloor}\ifc{^{\textrm{IT}}\lceil}\tar{e}\ifc{\rceil}\tar{\rfloor}\right]\rightarrow\mathcal{E}_{\tar{\Sigma}}\left[\tar{e}\right]}}
\end{mathpar}
\caption{The embedding \ensuremath{L_\text{IFC}(\alpha,\Red{\lambda}\;\!\!)}, where
\ensuremath{\Red{\lambda}\;\!\!\mathrel{=}(\tar{\Sigma},\tar{E},\tar{e},\tar{v},\rightarrow)}}
\label{fig:embedding}
\end{figure}

\begin{itemize}
    \item We extend the values, expressions and evaluation contexts of
      both languages to allow for terms in one language to
      be embedded in the other, as
      in~\cite{Matthews:2007:OSM:1190216.1190220}.
In the target language, an IFC expression appears as \ensuremath{\tar{_{\textrm{TI}}\lfloor}\ifc{\Varid{e}}\tar{\rfloor}} (``\Red{T}arget-outside,
IFC-inside''); in the IFC language, a target language expression appears as \ensuremath{\ifc{^{\textrm{IT}}\lceil}\tar{e}\ifc{\rceil}} ( ``{\color{blue}{I}}FC-outside, target-inside'').
    \item We reinterpret \ensuremath{\mathcal{E}} to be evaluation contexts on task lists, providing definitions for \ensuremath{\mathcal{E}_{\tar{\Sigma}}} and \ensuremath{\mathcal{E}_{\ifc{\Sigma}}^{\ifc{\Varid{i}},\ifc{l}}}.  These rules only operate on the first task in the task list, which by convention is the only task executing.
    \item We reinterpret \ensuremath{\rightarrow}, an operation on a single task, in terms of \ensuremath{\hookrightarrow}, operation on task lists.  The correspondence is simple: a task executes a step and then is rescheduled in the task list according to schedule policy \ensuremath{\alpha}.
    Fig.~\ref{fig:scheduler} defines two concrete schedulers.
    \item Finally, we define some rules for scheduling, handling sandboxing tasks (which interact with the state of the target language),
    and intermediating between the borders of the two languages.
\end{itemize}

\begin{figure}[t]
\begin{tabular}{ll}
\begin{minipage}{.45\textwidth}
\begin{hscode}\SaveRestoreHook
\column{B}{@{}>{\hspre}l<{\hspost}@{}}%
\column{3}{@{}>{\hspre}l<{\hspost}@{}}%
\column{40}{@{}>{\hspre}l<{\hspost}@{}}%
\column{E}{@{}>{\hspre}l<{\hspost}@{}}%
\>[3]{}\textsc{RR}_{{\tiny\mathrm{\text{step}}}}(\ifc{\Varid{t}}_{1},\ifc{\Varid{t}}_{2},\ldots){}\<[40]%
\>[40]{}\mathrel{=}\ifc{\Varid{t}}_{2},\ldots,\ifc{\Varid{t}}_{1}{}\<[E]%
\\
\>[3]{}\textsc{RR}_{{\tiny\mathrm{\text{done}}}}(\ifc{\Varid{t}}_{1},\ifc{\Varid{t}}_{2},\ldots){}\<[40]%
\>[40]{}\mathrel{=}\ifc{\Varid{t}}_{2},\ldots{}\<[E]%
\\
\>[3]{}\textsc{RR}_{{\tiny\mathrm{\text{noStep}}}}(\ifc{\Varid{t}}_{1},\ifc{\Varid{t}}_{2},\ldots){}\<[40]%
\>[40]{}\mathrel{=}\ifc{\Varid{t}}_{2},\ldots{}\<[E]%
\\
\>[3]{}\textsc{RR}_{{\tiny\mathrm{\text{sandbox}}}}(\ifc{\Varid{t}}_{1},\ifc{\Varid{t}}_{2},\ldots){}\<[40]%
\>[40]{}\mathrel{=}\ifc{\Varid{t}}_{2},\ldots,\ifc{\Varid{t}}_{1}{}\<[E]%
\ColumnHook
\end{hscode}\resethooks
\end{minipage} &
\begin{minipage}{.45\textwidth}
\begin{hscode}\SaveRestoreHook
\column{B}{@{}>{\hspre}l<{\hspost}@{}}%
\column{3}{@{}>{\hspre}l<{\hspost}@{}}%
\column{35}{@{}>{\hspre}l<{\hspost}@{}}%
\column{E}{@{}>{\hspre}l<{\hspost}@{}}%
\>[3]{}\textsc{Seq}_{{\tiny\mathrm{\text{step}}}}(\ifc{\Varid{t}}_{1},\ifc{\Varid{t}}_{2},\ldots){}\<[35]%
\>[35]{}\mathrel{=}\ifc{\Varid{t}}_{1},\ifc{\Varid{t}}_{2},\ldots{}\<[E]%
\\
\>[3]{}\textsc{Seq}_{{\tiny\mathrm{\text{noStep}}}}(\ifc{\Varid{t}}_{1},\ifc{\Varid{t}}_{2},\ldots){}\<[35]%
\>[35]{}\mathrel{=}\ifc{\Varid{t}}_{1},\ifc{\Varid{t}}_{2},\ldots{}\<[E]%
\\
\>[3]{}\textsc{Seq}_{{\tiny\mathrm{\text{done}}}}(\ifc{\Varid{t}}){}\<[35]%
\>[35]{}\mathrel{=}\ifc{\Varid{t}}{}\<[E]%
\\
\>[3]{}\textsc{Seq}_{{\tiny\mathrm{\text{done}}}}(\ifc{\Varid{t}}_{1},\ifc{\Varid{t}}_{2},\ldots){}\<[35]%
\>[35]{}\mathrel{=}\ifc{\Varid{t}}_{2},\ldots{}\<[E]%
\\
\>[3]{}\textsc{Seq}_{{\tiny\mathrm{\text{sandbox}}}}(\ifc{\Varid{t}}_{1},\ifc{\Varid{t}}_{2},\ldots,\ifc{\Varid{t}}_\Varid{n}){}\<[35]%
\>[35]{}\mathrel{=}\ifc{\Varid{t}}_\Varid{n},\ifc{\Varid{t}}_{1},\ifc{\Varid{t}}_{2},\ldots{}\<[E]%
\ColumnHook
\end{hscode}\resethooks
\end{minipage}
\end{tabular}
\vspace*{-0.5cm}
\caption{Scheduling policies (concurrent round robin on the left, sequential on the right).}
\label{fig:scheduler}
\end{figure}

The \textsc{I-sandbox} rule is used to create a new isolated task that
executes separately from the existing tasks (and can be communicated
with via \ensuremath{\mathbf{send}} and \ensuremath{\mathbf{recv}}).  When the new task is created, there
is the question of what the target language state of the new task should
be.  Our rule is stated generically in terms of a function \ensuremath{\kappa}.
Conservatively, \ensuremath{\kappa} may be simply thought of as the identity
function, in which
case the semantics of \ensuremath{\mathbf{sandbox}} are such that the state of the target language is \emph{cloned}
when sandboxing occurs.  However, this is not necessary: it is also valid for \ensuremath{\kappa}
to remove entries from the state.  In Section~\ref{sec:concrete}, we give a more detailed
discussion of the implications of the choice of \ensuremath{\kappa}, but all our
security claims will hold regardless of the choice of \ensuremath{\kappa}.

The rule \textsc{I-noStep} says something about configurations for which
it is not possible to take a transition.  The notation
$\ensuremath{\ifc{c}}\not\overset{\alpha}{\hookrightarrow}$ in the premise
is meant to be understood as
follows:  If the configuration \ensuremath{\ifc{c}} cannot take a step by any rule other
than \textsc{I-noStep}, then \textsc{I-noStep} applies and the
stuck task gets removed.

Rules \textsc{I-done} and \textsc{I-noStep} define the behavior of the system
when the current thread has reduced to a value, or gotten stuck, respectively.
While these definitions simply rely on the underlying scheduling policy
\ensuremath{\alpha} to modify the task list, as we describe in Sections~\ref{sec:formal}
and~\ref{sec:extensions}, these rules (notably, \textsc{I-noStep}) are crucial
to proving our security guarantees.
For instance, it is unsafe for the whole system to get stuck if a particular
task gets stuck, since a sensitive thread may then leverage this to leak
information through the termination channel.
Instead, as our example round-robin (\ensuremath{\textsc{RR}}) scheduler shows,
such tasks should simply be removed from the task list.
Many language runtime or Operating System schedulers implement such
schedulers.
Moreover, techniques such as instruction-based
scheduling~\cite{stefan:2013:eliminating, buiras:2013:a-library} can
be further applied close the gap between specified semantics and
implementation.

As in~\cite{Matthews:2007:OSM:1190216.1190220}, rules \textsc{T-border} and
\textsc{I-border} define the syntactic boundaries between the IFC and target
languages.
Intuitively, the boundaries respectively correspond to an upcall into and
downcall from the IFC runtime.
As an example, taking \ensuremath{\Red{\lambda_{\text{ES}}}} as the target language, we can now define a
blocking receive (inefficiently) in terms of the asynchronous \ensuremath{\mathbf{recv}} as series
of cross-language calls:
\begin{hscode}\SaveRestoreHook
\column{B}{@{}>{\hspre}l<{\hspost}@{}}%
\column{E}{@{}>{\hspre}l<{\hspost}@{}}%
\>[B]{}\mathbf{blockingRecv}\ \ifc{\Varid{x}}_{1},\ifc{\Varid{x}}_{2}\;\mathbf{in}\;\ifc{\Varid{e}}\triangleq{}\ifc{^{\textrm{IT}}\lceil}\mathbf{fix}\;(\lambda \Varid{k}.\tar{_{\textrm{TI}}\lfloor}\mathbf{recv}\ \ifc{\Varid{x}}_{1},\ifc{\Varid{x}}_{2}\;\mathbf{in}\;\ifc{\Varid{e}}\ \mathbf{else}\ \ifc{^{\textrm{IT}}\lceil}\Varid{k}\ifc{\rceil}\tar{\rfloor})\ifc{\rceil}{}\<[E]%
\ColumnHook
\end{hscode}\resethooks


For any target language \ensuremath{\Red{\lambda}} and scheduling policy \ensuremath{\alpha}, this
embedding defines an IFC language, which we will
refer to as \ensuremath{L_\text{IFC}(\alpha,\Red{\lambda})}.
\section{Security Guarantees}
\label{sec:formal}

We are interested in proving non-interference about many programming
languages.  This requires an appropriate definition of this notion
that is language agnostic, so in this section, we present a few
general definitions for what an information flow control language is
and what non-interference properties it may have.
In particular, we show that \ensuremath{L_\text{IFC}(\alpha,\Red{\lambda})}, with an appropriate
scheduler \ensuremath{\alpha}, satisfies non-interference~\cite{Goguen82},
without making any reference to properties of \ensuremath{\Red{\lambda}}.
We state the appropriate theorems here, and provide the formal
proofs in
\appref{sec:app:proof}.

\cut{
The precise scheduling policy dictates
what guarantee we can achieve for programs with diverging tasks.
For a sequential scheduler \ensuremath{\textsc{Seq}}, we will only be able to show
\emph{termination-insensitive non-interference}, where a program
may diverge based on secret data; this allows attackers to observe
secrets by observing the termination of tasks.
For the concurrent round-robin schedule \ensuremath{\textsc{RR}},
we can show a stronger result known as
\emph{termination-sensitive non-interference},
where termination attacks cannot leak information.
}

\subsection{Erasure Function}

When defining the security guarantees of an information flow control,
we must characterize what the \emph{secret inputs} of a program are.  Like
other work~\cite{Li+:2010:arrows,Russo+:Haskell08,lio,stefan:addressing-covert},
we specify and prove non-interference using \emph{term erasure}.
Intuitively, term erasure allows us to show that an attacker does not learn
any sensitive information from a program if the program behaves identically
(from the attackers point of view) to a program with all sensitive data
``erased''.
To interpret a language under information flow control, we define a function \ensuremath{\varepsilon_{l}} that
performs erasures by mapping configurations to erased configurations,
usually by rewriting (parts of) configurations that are more sensitive
than \ensuremath{l} to a new syntactic construct \ensuremath{\bullet}.  We define
an information flow control language as follows:

\begin{definition}[Information flow control language]
    An information flow control language \ensuremath{\textbf{L}} is a tuple \ensuremath{(\Delta,\hookrightarrow,\varepsilon_{l})}, where $\ensuremath{\Delta}$ is the type of machine configurations (members
    of which are usually denoted by the metavariable \ensuremath{c}), \ensuremath{\hookrightarrow} is a
    reduction relation between machine configurations and \ensuremath{\varepsilon_{l}\mathbin{:}\Delta\rightarrow\varepsilon(\Delta)}
    is an erasure function parametrized on labels from machine configurations to \emph{erased} machine
    configurations \ensuremath{\varepsilon(\Delta)}.  Sometimes, we use \ensuremath{V} to refer to set of
     terminal configurations in \ensuremath{\Delta}, i.e., configurations where
     no further transitions are possible.
\end{definition}

Our language \ensuremath{L_\text{IFC}(\alpha,\Red{\lambda}\;\!\!)} fulfills
this definition as \ensuremath{(\ifc{\Delta},\overset{\alpha}{\hookrightarrow},\varepsilon_{\ifc{l}})}, where
$\ensuremath{\ifc{\Delta}} = \ensuremath{\ifc{\Sigma}} \times \operatorname{List}(\ensuremath{\ifc{\Varid{t}}})$.  The set of terminal conditions
$\ensuremath{\ifc{V}}$ is $\ensuremath{\ifc{\Sigma}} \times \ensuremath{\ifc{\Varid{t}}}_V$, where $\ensuremath{\ifc{\Varid{t}}}_V \subset \ensuremath{\ifc{\Varid{t}}}$ is the
type for tasks whose expressions have been reduced to
values.\footnote{
  Here, we abuse notation by describing types for configuration parts using the
  same metavariables as the ``instance'' of the type, e.g., \ensuremath{\ifc{\Varid{t}}} for the type of
  task.
}
The erased configuration \ensuremath{\varepsilon(\ifc{\Delta})} extends \ensuremath{\ifc{\Delta}} with configurations
containing \ensuremath{\bullet}, and Fig.~\ref{fig:erasure} gives the precise definition for
our erasure function \ensuremath{\varepsilon_{\ifc{l}}}.
Essentially, a task and its corresponding message queue is completely erased from the task
list if its label does not flow to the attacker observation level \ensuremath{\ifc{l}}.
Otherwise, we apply the erasure function homomorphically and remove any messages
from the task's message queue that are more sensitive than \ensuremath{\ifc{l}}.

\begin{figure} 
\begin{align*}
  &\ensuremath{\varepsilon_{\ifc{l}}(\ifc{\Sigma};\ifc{\Varid{ts}})\mathrel{=}\varepsilon_{\ifc{l}}(\ifc{\Sigma});\text{filter}\;(\lambda \ifc{\Varid{t}}.\ifc{\Varid{t}}\mathrel{=}\bullet)\;(\text{map}\;\varepsilon_{\ifc{l}}\;\ifc{\Varid{ts}})}\\
  &\ensuremath{\varepsilon_{\ifc{l}}(\langle \tar{\Sigma}, \ifc{\Varid{e}}\rangle^{\ifc{\Varid{i}}}_{\ifc{l}'})\mathrel{=}} \begin{cases}
    \ensuremath{\bullet} & \ensuremath{\ifc{l}'\;\not\flows{}\;\ifc{l}} \\
    \ensuremath{\langle \varepsilon_{\ifc{l}}(\tar{\Sigma}), \varepsilon_{\ifc{l}}(\ifc{\Varid{e}})\rangle^{\ifc{\Varid{i}}}_{\ifc{l}'}} & \text{otherwise}
  \end{cases} \\
  &\ensuremath{\varepsilon_{\ifc{l}}(\ifc{\Sigma}\left[\ifc{\Varid{i}}\mapsto{}\ifc{\Theta}\right])\mathrel{=}} \begin{cases}
    \ensuremath{\varepsilon_{\ifc{l}}(\ifc{\Sigma})} & \text{\ensuremath{\ifc{l}'\;\not\flows{}\;\ifc{l}}, where \ensuremath{\ifc{l}'} is the label of thread \ensuremath{\ifc{\Varid{i}}}}\\
    \ensuremath{\varepsilon_{\ifc{l}}(\ifc{\Sigma})\left[\ifc{\Varid{i}}\mapsto{}\varepsilon_{\ifc{l}}(\ifc{\Theta})\right]} & \text{otherwise}
  \end{cases} \\
  &\ensuremath{\varepsilon_{\ifc{l}}(\ifc{\Theta})\mathrel{=}\ifc{\Theta}\preceq\ifc{l}}  \qquad\qquad\qquad  \ensuremath{\varepsilon_{\ifc{l}}(\emptyset)\mathrel{=}\emptyset}
\end{align*}
\vspace*{-0.8cm}
\caption{ Erasure function for tasks, queue maps, message queues, and
configurations.  In all other cases, including target-language constructs,
\ensuremath{\varepsilon_{\ifc{l}}} is applied homomorphically.
Note that \ensuremath{\varepsilon_{\ifc{l}}(\ifc{\Varid{e}})} is always equal to \ensuremath{\ifc{\Varid{e}}} (and similar for \ensuremath{\tar{\Sigma}})
in this simple setting.  However,
when the IFC language is extended with more constructs as shown in
Section~\ref{sec:extensions}, then this will no longer be the case.
\label{fig:erasure} }
\end{figure}

The definition of an erasure function is quite important: it captures
the attacker model, stating what can and cannot be observed by the attacker.
In our case, we assume that the attacker cannot observe sensitive tasks or
messages, or even the number of such entities.
While such assumptions are standard~\cite{Castellani:Boudol:ICALP01,
stefan:addressing-covert}, our definitions allow for
stronger attackers that may be able to inspect resource usage.\footnote{
  We believe that we can extend \ensuremath{L_\text{IFC}(\alpha,\Red{\lambda}\;\!\!)} to
  such models using the resource limits techniques of~\cite{yangresource}.
  We leave this extension to future work.
}

\subsection{Non-Interference}

Given an information flow control language, we can now define non-interference.
Intuitively, we want to make statements about the attacker's observational
power at some security level \ensuremath{l}.  This is done by defining an equivalence
relation called \ensuremath{l}-equivalence on configurations: an attacker should
not be able to distinguish two configurations that are \ensuremath{l}-equivalent.
Since our erasure function captures what an attacker can or cannot observe, we simply define this
equivalence as the syntactic-equivalence of erased configurations~\cite{stefan:addressing-covert}.
\begin{definition}[\ensuremath{l}-equivalence]
    In a language \ensuremath{(\Delta,\hookrightarrow,\varepsilon_{l})}, two machine configurations
    \ensuremath{c,c'\in\Delta} are considered $l$-equivalent, written as \ensuremath{c\approx_lc'},
    if \ensuremath{\varepsilon_{l}(c)\mathrel{=}\varepsilon_{l}(c')}.
\end{definition}

We can now state that a language satisfies non-interference if an
attacker at level \ensuremath{l} cannot distinguish the runs of any two \ensuremath{l}-equivalent
configurations.
This particular property is called termination sensitive non-interference
(TSNI).  Besides the obvious requirement to not leak secret information
to public channels, this definition also requires the termination
of public tasks to be independent of secret tasks.
Formally, we define TSNI as follows:

\begin{definition}[Termination Sensitive Non-Interference (TSNI)]
  A language \ensuremath{(\Delta,\hookrightarrow,\varepsilon_{l})} satisfies termination
  sensitive non-interference if for any label \ensuremath{l}, and configurations
  $\ensuremath{c_{1},c_{1}',c_{2}}\in\ensuremath{\Delta}$, if
  \begin{equation} \label{eq:tsni-lhs}
    \ensuremath{c_{1}} \approx_{\ensuremath{l}} \ensuremath{c_{2}}
    \qquad \text{and} \qquad
    \ensuremath{c_{1}} \ensuremath{\hookrightarrow}^* \ensuremath{c_{1}'}
  \end{equation}
  then there exists a configuration $\ensuremath{c_{2}'}\in\ensuremath{\Delta}$ such that
  \begin{equation} \label{eq:tsni-rhs}
    \ensuremath{c_{1}'} \approx_{\ensuremath{l}} \ensuremath{c_{2}'}
     \qquad \text{and} \qquad
    \ensuremath{c_{2}} \ensuremath{\hookrightarrow}^* \ensuremath{c_{2}'}
    \ \text{.}
  \end{equation}
\end{definition}
In other words, if we take two \ensuremath{l}-equivalent configurations, then for every
intermediate step taken by the first configuration, there is a corresponding
number of steps that the second configuration can take to result in a
configuration that is \ensuremath{l}-equivalent to the first resultant configuration.
By symmetry, this applies to all intermediate steps from the second configuration
as well.
\ifextended
We remark that this notion of non-interfernce is similar to
\emph{progress sensitive non-interference (PSNI)}, which accounts for
leakage via progress (or termination) channels, as used for static
systems~\cite{moore2012precise}.
\fi

Our language satisfies TSNI 
\ifextended
(and thus PSNI)
\fi
under the round-robin scheduler
\ensuremath{\textsc{RR}} of Fig.~\ref{fig:scheduler}.
\begin{theorem}[Concurrent IFC language is TSNI]
  \label{thm:rr-tsni}
For any target language \ensuremath{\Red{\lambda}}, \ensuremath{L_\text{IFC}(\textsc{RR},\Red{\lambda})} satisfies TSNI.
\end{theorem}

In general, however, non-interference will not hold for an arbitrary scheduler \ensuremath{\alpha}.
For example, \ensuremath{L_\text{IFC}(\alpha,\Red{\lambda})} with a scheduler that inspects a
sensitive task's current state when deciding which task to schedule next
will in general break non-interference~\cite{russo2006securing,BartheRRS07}.
%
%

However, even non-adversarial schedulers are not always safe.
Consider, for example, the sequential scheduling policy \ensuremath{\textsc{Seq}} given in
Fig.~\ref{fig:scheduler}.
It is easy to show that \ensuremath{L_\text{IFC}(\textsc{Seq},\Red{\lambda})} does not satisfy
TSNI:
consider a target language similar to \ensuremath{\Red{\lambda_{\text{ES}}}} with an
additional expression terminal \ensuremath{\tar{\Uparrow}} that denotes a divergent computation,
i.e., \ensuremath{\tar{\Uparrow}} always reduces to \ensuremath{\tar{\Uparrow}} and a simple label lattice \ensuremath{\{\mskip1.5mu \mathsf{pub},\mathsf{sec}\mskip1.5mu\}} such that \ensuremath{\mathsf{pub}\flows{}\mathsf{sec}}, but \ensuremath{\mathsf{sec}\not\flows{}\mathsf{pub}}.
Consider the following two configurations in this language:
\begin{hscode}\SaveRestoreHook
\column{B}{@{}>{\hspre}l<{\hspost}@{}}%
\column{71}{@{}>{\hspre}l<{\hspost}@{}}%
\column{E}{@{}>{\hspre}l<{\hspost}@{}}%
\>[B]{}\ifc{c}_{1}\mathrel{=}\ifc{\Sigma};\langle \tar{\Sigma}_{1}, \ifc{^{\textrm{IT}}\lceil}\;\mathbf{if}\;\mathbf{false}\;\mathbf{then}\;\tar{\Uparrow}\;\mathbf{else}\;\mathbf{true}\ifc{\rceil}\rangle^{\mathrm{1}}_{\mathsf{sec}},{}\<[71]%
\>[71]{}\langle \tar{\Sigma}_{2}, \ifc{\Varid{e}}\rangle^{\mathrm{2}}_{\mathsf{pub}}{}\<[E]%
\\
\>[B]{}\ifc{c}_{2}\mathrel{=}\ifc{\Sigma};\langle \tar{\Sigma}_{1}, \ifc{^{\textrm{IT}}\lceil}\;\mathbf{if}\;\mathbf{true}\;\mathbf{then}\;\tar{\Uparrow}\;\mathbf{else}\;\mathbf{true}\ifc{\rceil}\rangle^{\mathrm{1}}_{\mathsf{sec}},{}\<[71]%
\>[71]{}\langle \tar{\Sigma}_{2}, \ifc{\Varid{e}}\rangle^{\mathrm{2}}_{\mathsf{pub}}{}\<[E]%
\ColumnHook
\end{hscode}\resethooks
These two configurations are \ensuremath{\mathsf{pub}}-equivalent, but \ensuremath{\ifc{c}_{1}} will reduce
(in two steps)
to \ensuremath{\ifc{c}_{1}'\mathrel{=}\ifc{\Sigma};\langle \tar{\Sigma}_{1}, \ifc{^{\textrm{IT}}\lceil}\mathbf{true}\ifc{\rceil}\rangle^{\mathrm{2}}_{\mathsf{pub}}}, whereas \ensuremath{\ifc{c}_{2}} will not make
any progress.
Suppose that \ensuremath{\ifc{\Varid{e}}} is a computation that writes to a \ensuremath{\mathsf{pub}} channel,\footnote{
Though we do not model labeled channels, extending the calculus with such a
feature is straightforward, see Section~\ref{sec:extensions}.}
then the \ensuremath{\mathsf{sec}} task's decision to diverge or not is directly leaked to a
public entity.

%
To accommodate for sequential languages, or cases where a weaker guarantee
is sufficient, we consider an alternative non-interference property called termination insensitive
non-interference (TINI).  This property can also be upheld by sequential languages at the cost
of leaking through (non)-termination~\cite{Askarov:2008}.
\begin{definition}[Termination insensitive non-interference (TINI)]
  A language \ensuremath{(\Delta,V,\hookrightarrow,\varepsilon_{l})} is termination
  insensitive non-interfering if for any label \ensuremath{l}, and configurations
  $\ensuremath{c_{1},c_{2}}\in\ensuremath{\Delta}$ and $\ensuremath{c_{1}',c_{2}'}\in\ensuremath{V}$, it holds that
  \[
    (\ensuremath{c_{1}} \approx_{\ensuremath{l}} \ensuremath{c_{2}}
    \;\;\land\;\;
    \ensuremath{c_{1}} \ensuremath{\hookrightarrow}^* \ensuremath{c_{1}'}
    \;\;\land\;\;
    \ensuremath{c_{2}} \ensuremath{\hookrightarrow}^* \ensuremath{c_{2}'})
    \implies
    \ensuremath{c_{1}'} \approx_{\ensuremath{l}} \ensuremath{c_{2}'}
  \]
\end{definition}

TINI states that if we take two \ensuremath{l}-equivalent configurations, and both
configurations reduce to final configurations (i.e.,
configurations for which there are no
possible further transitions), then the end configurations are also
\ensuremath{l}-equivalent.
We highlight that this statement is much weaker than TSNI: it only states that
terminating programs do not leak sensitive data, but makes no statement
about non-terminating programs.

As shown by compilers~\cite{jif,FlowCaml}, interpreters~\cite{JSFlow}, and
libraries~\cite{Russo+:Haskell08,lio}, TINI is useful for sequential
settings. In our case, we show that our IFC language with the sequential scheduling policy
\ensuremath{\textsc{Seq}} satisfies TINI.
\begin{theorem}[Sequential IFC language is TINI]
  \label{thm:seq-tini}
For any target language \ensuremath{\Red{\lambda}}, \ensuremath{L_\text{IFC}(\textsc{Seq},\Red{\lambda})} satisfies TINI.
\end{theorem}

\section{Isomorphisms and Restrictions}
\label{sec:concrete}

\newcommand{\con}[1]{\ensuremath{{\color{red} #1}}}
\newcommand{\abs}[1]{\ensuremath{{\color{blue} #1}}}

The operational semantics we have defined in the previous section
satisfy non-interference by design.
We achieve this general statement that works for a large class of
languages by having different tasks executing completely isolated from
each other, such that every task has its own state.
In some cases, this strong separation is desirable, or even necessary.
Languages like C provide direct access to memory locations without
mechanisms in the language to achieve a separation of the heap.
On the other hand, for other languages, this
strong isolation of tasks can be
undesirable, e.g., for performance reasons.
For instance, for the language \ensuremath{\Red{\lambda_{\text{ES}}}}, our presentation so far
requires a separate heap per task, which is not very practical.
Instead, we would like to
more tightly couple the integration of the target and IFC
languages by reusing existing infrastructure.  In the running example,
a concrete implementation might use a single global heap.
More precisely, instead of using a configuration of the form
$\ensuremath{\ifc{\Sigma};\langle \tar{\Sigma}_{1}, \ifc{\Varid{e}}_{1}\rangle^{\ifc{\Varid{i}}_{1}}_{\ifc{l}_{1}},\langle \tar{\Sigma}_{2}, \ifc{\Varid{e}}_{2}\rangle^{\ifc{\Varid{i}}_{2}}_{\ifc{l}_{2}}\;\ldots}$
we would like a single global heap as in
$\ensuremath{\ifc{\Sigma};\tar{\Sigma};\langle \ifc{\Varid{e}}_{1}\rangle^{\ifc{\Varid{i}}_{1}}_{\ifc{l}_{1}},\langle \ifc{\Varid{e}}_{2}\rangle^{\ifc{\Varid{i}}_{2}}_{\ifc{l}_{2}},\ldots}$

If the operational rules are adapted na\"ively to this new setting,
then non-interference can be violated: as we mentioned earlier,
shared mutable cells could be used to leak sensitive information.
What we would like is a way of characterizing safe modifications to
the semantics which preserve non-interference.
The intention of our single heap implementation is to permit
efficient execution while \emph{conceptually maintaining isolation between
tasks} (by not allowing sharing of references between them).
This intuition of having a different (potentially more efficient)
concrete semantics that behaves like the abstract semantics
can be formalized by the following definition:


\begin{definition}[Isomorphism of information flow control languages]
  A language \ensuremath{(\Delta,\hookrightarrow,\varepsilon_{l})} is \textit{isomorphic} to a
  language \ensuremath{(\Delta',\hookrightarrow',\varepsilon'_{l})} if there exist total functions \ensuremath{\Varid{f}\mathbin{:}\Delta\rightarrow\Delta'} and \ensuremath{\Varid{f}^{-1}\mathbin{:}\Delta'\rightarrow\Delta} such that \ensuremath{\Varid{f}\circ\Varid{f}^{-1}\mathrel{=}id_{\Delta}} and \ensuremath{\Varid{f}^{-1}\circ\Varid{f}\mathrel{=}id_{\Delta'}}.  Furthermore, \ensuremath{\Varid{f}} and \ensuremath{\Varid{f}^{-1}} are functorial (e.g., if
  $x'\ R'\ y'$ then $f(x')\ R\ f(y')$) over both
  $l$-equivalences and \ensuremath{\hookrightarrow}.

  If we weaken this restriction such that \ensuremath{\Varid{f}^{-1}} does
  not have to be functorial over \ensuremath{\hookrightarrow}, we call the
  language \ensuremath{(\Delta,\hookrightarrow,\varepsilon_{l})} \textit{weakly isomorphic} to
  \ensuremath{(\Delta',\hookrightarrow',\varepsilon'_{l})}.
\end{definition}

Providing an isomorphism between the two languages allows us to
preserve (termination sensitive or insensitive) non-interference
as the following two theorems state.

\begin{theorem}[Isomorphism preserves TSNI]
  \label{thm:iso-tsni}
  If $L$ is isomorphic to $L'$ and $L'$ satisfies TSNI, then
  $L$ satisfies TSNI.
\end{theorem}

\begin{proof}
  Shown by transporting configurations and reduction derivations from
  $L$ to $L'$, applying TSNI, and then transporting the
  resulting configuration, $l$-equivalence and multi-step derivation back.
  \qed
\end{proof}

Only weak isomorphism is necessary for TINI. Intuitively, this is because
it is not necessary to back-translate reduction sequences in $L'$ to
$L$; by the definition of TINI, we have both reduction sequences in $L$
by assumption.

\begin{theorem}[Weak isomorphism preserves TINI]
  \label{thm:iso-tini}
  If a language $L$ is weakly isomorphic to a language $L'$, and $L'$
  satisfies TINI, then $L$ satisfies TINI.
\end{theorem}

\begin{proof}
  Shown by transporting configurations and reduction derivations
  from $L$ to $L'$, applying TINI and transporting the resulting
  equivalence back using functoriality of \ensuremath{\Varid{f}^{-1}} over $l$-equivalences.
  \qed
\end{proof}

Unfortunately, an isomorphism is often too strong of a requirement.
To obtain an isomorphism with our single heap semantics, we need to mimic the
behavior of several heaps with a single actual heap.
The interesting cases are when we sandbox
an expression and when messages are sent and received.
The rule for sandboxing is
parametrized by the strategy \ensuremath{\kappa} (see Section~\ref{sec:retrofit}),
which defines what heap the new task
should execute with.  We have considered two choices:

\begin{itemize}
    \item When we sandbox into an empty heap, existing addresses
in the sandboxed expression are no longer valid and the
task will get stuck (and then removed by \textsc{I-noStep}).
Thus, we must rewrite the sandboxed expression so that
all addresses point to fresh addresses
guaranteed to not occur in the heap.  Similarly,
sending a memory address should be rewritten.

\item When we clone the heap, we have to copy everything
reachable from the sandboxed expression and replace all addresses
correspondingly.  Even worse, the behavior of sending a memory address
now depends on whether that address existed at the time the receiving
task was sandboxed;  if it did, then the address should be rewritten to the
existing one.
\end{itemize}

Isomorphism demands we implement this convoluted behavior,
despite our initial motivation of a more efficient implementation.

\subsection{Restricting the IFC Language}

A better solution is to forbid sandboxed expressions as well
as messages sent to other tasks to contain memory addresses in the
first place.  In a statically typed language, the type system could
prevent this from happening.
In dynamically typed languages such as \ensuremath{\Red{\lambda_{\text{ES}}}}, we might
restrict the transition for \ensuremath{\mathbf{sandbox}} and \ensuremath{\mathbf{send}} to only allow expressions
without memory addresses.

While this sounds plausible, it is worth noting that we are modifying the IFC language semantics,
which raises the question of whether non-interference is preserved.
This question can be subtle: it is easy to remove a transition from
a language and invalidate TSNI.  Intuitively
if the restriction depends on secret data, then a public thread
can observe if some other task terminates or not, and from that obtain
information about the secret data that was used to restrict the
transition.
With this in mind, we require semantic rules to get restricted only
based on information observable by the task triggering them.
This ensures that non-interference is preserved, as the
restriction does not depend on confidential information.
Below, we give the formal definition of this condition for the
abstract IFC language \ensuremath{L_\text{IFC}(\alpha,\Red{\lambda})}.







\begin{definition}[Restricted IFC language]
  \label{def:restricted}
  For a family of predicates $\mathcal P$ (one for every reduction
  rule), we call
  \ensuremath{L_\text{IFC}^{\mathcal{P}}(\alpha,\Red{\lambda})} a restricted IFC language
  if its definition is equivalent to the abstract language
  \ensuremath{L_\text{IFC}(\alpha,\Red{\lambda})}, with the following exception:
  the reduction rules are restricted
  by adding a predicate $P \in \mathcal P$ to the premise of
  all rules other than \textsc{I-noStep}.  Furthermore, the predicate $P$
  can depend only on the \textit{erased} configuration
  \ensuremath{\varepsilon_{\ifc{l}}(\ifc{c})}, where \ensuremath{\ifc{l}} is the label of the first task
  in the task list and \ensuremath{\ifc{c}} is the full configuration.
\end{definition}

By the following theorem, the restricted IFC language with an
appropriate scheduling policy is non-interfering.

\begin{theorem}
  \label{thm:restricted}
  For any target language \ensuremath{\Red{\lambda}} and family of predicates
  $\mathcal{P}$, the restricted IFC language \ensuremath{L_\text{IFC}^{\mathcal{P}}(\textsc{RR},\Red{\lambda})}
  is TSNI.  Furthermore, the IFC language
  \ensuremath{L_\text{IFC}^{\mathcal{P}}(\textsc{Seq},\Red{\lambda})} is TINI.
\end{theorem}

In~\appref{sec:single-heap} we give an example how this formalism can be used
to show non-intereference of an implementation of IFC with a single heap.
\section{Real World Languages}
\label{sec:real}

Our approach can be used to retrofit any language for which we
can achieve isolation with information flow control.
Unfortunately, controlling the external effects of a real-world
language, as to achieve isolation, is language-specific and varies
from one language to another.\footnote{
  Though we apply our framework to several real-world languages, it is
  conceivable that there are languages for which isolation cannot be
  easily achieved.
}
Indeed, even for a single language (e.g., JavaScript), how one
achieves isolation may vary according to the language runtime or
embedding (e.g., server and browser).

In this section, we describe several implementations and their
approaches to isolation.
In particular, we describe two JavaScript IFC implementations
building on the theoretical foundations of this work.
Then, we consider how our formalism could be applied to the C
programming language and connect it to a previous IFC system for
Haskell.
%

%

\subsection{JavaScript}
\label{sec:real:js}

JavaScript, as specified by
ECMAScript~\cite{ecma}, does not have any built-in
functionality for I/O.
%
For this language, which we denote by \ensuremath{\Red{\lambda_{\text{JS}}}}, the IFC system
\ensuremath{L_\text{IFC}(\textsc{RR},\Red{\lambda_{\text{JS}}})} can be implemented by exposing IFC primitives
to JavaScript as part of the runtime, and running multiple instances
of the JavaScript virtual machine in separate OS-level threads.
Unfortunately, this becomes very costly when a system, such as a
server-side web application, relies on many tasks.
%

Luckily, this issue is not unique to our work---browser layout engines
also rely on isolating code executing in separate iframes (e.g., according to the
same-origin policy).
Since creating an OS thread for each iframe is expensive, both
the V8 and SpiderMonkey JavaScript engines provide means for running
JavaScript code in isolation within a single OS thread,
on disjoint sub-heaps.
In V8, this unit of isolation is called a \emph{context}; in
SpiderMonkey, it is called a \emph{compartment}.
(We will use these terms interchangeably.)
Each context is associated with a global object, which, by
default, implements the JavaScript standard library (e.g.,
\text{\tt Object}, \text{\tt Array}, etc.).
Naturally, we adopt contexts to implement our notion of tasks.

When JavaScript is embedded in browser layout engines,
or in server-side platforms such as Node.js,
additional APIs such as the Document Object Model (DOM) or the file
system get exposed as part of the runtime system.
These features are exposed by extending the global object, just like
the standard library.  For this reason, it is easy to modify
these systems to forbid external effects when implementing
an IFC system, ensuring that important effects can be reintroduced in a safe manner.

\begin{figure}[t]
\centerline{\includegraphics[width=0.7\columnwidth]{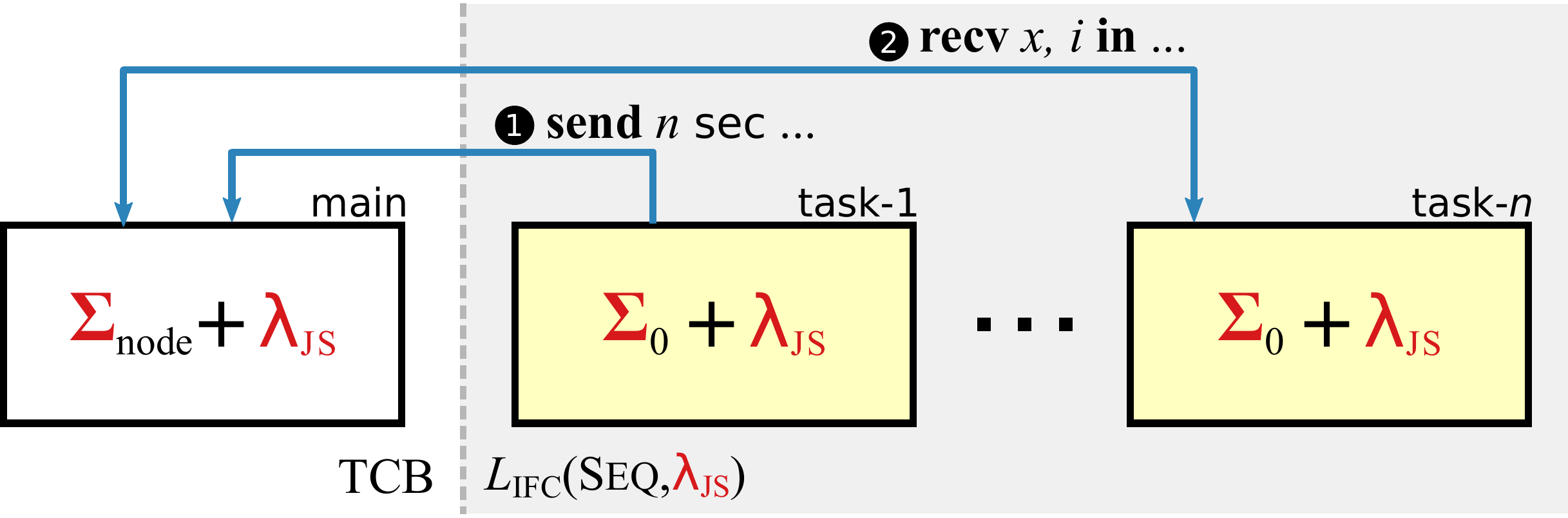}}
\caption{\label{fig:node}
This example shows how our trusted monitor (left) is used to mediate
communication between two tasks for which IFC is enforced (right).}
\end{figure}
\paragraph{Server-side IFC for Node.js:}
We have implemented \ensuremath{L_\text{IFC}(\textsc{Seq},\Red{\lambda_{\text{JS}}})} for Node.js in the form of a
library, without modifying Node.js or the V8 JavaScript engine.
Our implementation\footnote{Available at \codelink{}.} provides a
library for creating new tasks, i.e., contexts whose global object
only contains the standard JavaScript library and our IFC primitives
(e.g., \ensuremath{\mathbf{send}} and \ensuremath{\mathbf{sandbox}}).
When mapped to our formal treatment, \ensuremath{\mathbf{sandbox}} is defined with \ensuremath{\kappa(\tar{\Sigma})\mathrel{=}\tar{\Sigma}_{0}}, where \ensuremath{\tar{\Sigma}_{0}} is the global object corresponding to the
standard JavaScript library and our IFC primitives.
%
%
These IFC operations are mediated by the trusted library code (executing
as the main Node.js context), which tracks the state (current label, messages,
etc.) of each task.  An example for \ensuremath{\mathbf{send}}/\ensuremath{\mathbf{recv}} is shown in
Fig.~\ref{fig:node}.
Our system conservatively restricts the kinds of messages that can be
exchanged, via \ensuremath{\mathbf{send}} (and \ensuremath{\mathbf{sandbox}}), to string values.
In our formalization, this amounts to restricting the IFC language rule
for \ensuremath{\mathbf{send}} in the following way:
\newcommand{\str}{"string"}
\begin{mathpar}
\inferrule[JS-send]
{
\ensuremath{\ifc{l}\;\flows{}\;\ifc{l}'}\\
\ensuremath{\ifc{\Sigma}\;(\ifc{\Varid{i}}')\mathrel{=}\ifc{\Theta}}\\
\ensuremath{\ifc{\Sigma}'\mathrel{=}\ifc{\Sigma}\;[\mskip1.5mu \ifc{\Varid{i}}'\mapsto{}(\ifc{l}',\ifc{\Varid{i}},\ifc{\Varid{v}}),\ifc{\Theta}\mskip1.5mu]}\\
\ensuremath{\ifc{\Varid{e}}\mathrel{=}\ifc{^{\textrm{IT}}\lceil}\tar{e}\ifc{\rceil}}\\
\ensuremath{\mathcal{E}_{\tar{\Sigma}}\left[\texttt{typeOf}(\tar{e})\texttt{ === \str}\right]\rightarrow\mathcal{E}_{\tar{\Sigma}}\left[\texttt{true}\right]}
}
{\ensuremath{\ifc{\Sigma};\langle \tar{\Sigma}, \ifc{E}\ifc{[}\mathbf{send}\;\ifc{\Varid{i}}'\;\ifc{l}'\;\ifc{\Varid{v}}\ifc{]}_\ifc{I}\rangle^{\ifc{\Varid{i}}}_{\ifc{l}},\ldots\hookrightarrow\ifc{\Sigma}';\alpha_{{\tiny\mathrm{\text{step}}}}(\langle \tar{\Sigma}, \ifc{E}\ifc{[}\langle\rangle\ifc{]}_\ifc{I}\rangle^{\ifc{\Varid{i}}}_{\ifc{l}},\ldots)}}
\end{mathpar}
Of course, we provide a convenience library which marshals JSON
objects to/from strings.
We remark that this is not unlike existing message-passing JavaScript
APIs, e.g., \texttt{postMessage}, which impose similar restrictions as
to avoid sharing references between concurrent code.

While the described system implements \ensuremath{L_\text{IFC}(\textsc{Seq},\Red{\lambda_{\text{JS}}})}, applications
typically require access to libraries (e.g., the file system library
\textsf{fs}) that have external effects.
Exposing the Node.js APIs directly to sandboxed tasks is unsafe.
Instead, we implement libraries (like a labeled version of \textsf{fs}) as
message exchanges between the sandboxed tasks (e.g., \textsf{task-1}
in Fig.~\ref{fig:node}) and the main Node.js task that implements
the IFC monitor.
While this is safer than simply wrapping unsafe objects, which can
potentially be exploited to access objects outside the context (e.g.,
as seen with
\ifextended
ADSafe, FBJS, and Caja~\cite{taly2011automated, maffeis2010object, maffeis2009language}),
\else
ADSafe~\cite{taly2011automated}),
\fi
adding features such as the
\textsf{fs} requires the code in the main task to ensures that labels
are properly propagated and enforced.
Unfortunately, while imposing such a proof burden
is undesirable, this also has to be expected:
different language environments expose different libraries for
handling external I/O, and the correct treatment of external effects
is application specific.
%
%
%
We do not extend our formalism to account for the  particular
interface to the file system, HTTP client, etc., as this is
specific to the Node.js implementation and does not generalize
to other systems.

\paragraph{Client-side IFC:}
This work provides the formal basis for the core part of the COWL
client-side JavaScript IFC system~\cite{swapi}.
Like our Node.js implementation, COWL takes a coarse-grained approach
to providing IFC for JavaScript programs.
%
However, COWL's IFC monitor is implemented in
the browser layout engine instead (though still leaving the JavaScript engine
unmodified).

Furthermore, COWL repurposes existing contexts (e.g., iframes and
pages) as IFC tasks, only imposing additional constraints on how they
communicate.
As with Node.js, at its core, the global object of a COWL task
should only contain the standard JavaScript libraries and
\texttt{postMessage}, whose semantics are modeled by our
\textsc{JS-send} rule.
However, existing contexts have objects such as the DOM, which
require COWL to restrict a task's external effects.
To this end, COWL mediates any communication (even via the DOM) at
the context boundary. 

Simply disallowing all the external effects is overly-restricting for
real-world applications (e.g., pages typically load images, perform
network requests, etc.).
In this light, COWL allows safe network communication by associating an
implicit label with remote hosts (a host's label corresponds to
its origin).
In turn, when a task performs a request, COWL's IFC monitor
ensures that the task label can flow to the remote origin label.
While the external effects of COWL can be formally modeled, we do not
model them in our formalism, 
since, like for the
Node.js case, they are specific to this system.
%

\subsection{Haskell}
\label{sec:real:hs}
Our work borrows ideas from the LIO Haskell coarse-grained IFC
system~\cite{lio, stefan:addressing-covert}.
LIO relies on Haskell's type system and monadic encoding of
effects to achieve isolation and define the IFC sub-language.
Specifically, LIO provides the \text{\tt LIO} monad as a way of restricting
(almost all) side-effects.
In the context of our framework, LIO can be understood as follows: the
\emph{pure subset} of Haskell is the target language, while the
monadic subset of Haskell, operating in the \text{\tt LIO} monad, is the
IFC language.

%
Unlike our proposal, LIO originally associated labels with exceptions, in a
similar style to fine-grained
systems~\cite{stefan:2012:arxiv-flexible, Hritcu:2013:YIB:2497621.2498098}.
In addition to being overly complex, the interaction of exceptions
with clearance (which sets an upper bound on the floating label, see
\appref{sec:clearance}) was incorrect: the clearance
was restored to the clearance at point of the catch.
Furthermore, pure exceptions (e.g., divide by zero) always percolated to
trusted code, effectively allowing for denial of service attacks.
The insights gained when viewing coarse-grained IFC as presented in this
paper led to a much cleaner, simpler treatment of exceptions,
which has now been adopted by LIO.

\subsection{C}
\label{sec:real:c}
C programs are able to execute arbitrary (machine) code, access
arbitrary memory, and perform arbitrary system calls.
Thus, the confinement of C programs must be imposed by the underlying OS
and hardware.
For instance, our notion of isolation can be achieved using Dune's
hardware protection mechanisms~\cite{Belay:2012:DSU:2387880.2387913},
similar to 
Wedge~\cite{Belay:2012:DSU:2387880.2387913,
Bittau:2008:WSA:1387589.1387611}, but using an information flow control
policy.
Using page tables, a (trusted) IFC runtime could ensure that each task,
implemented as a lightweight process, can only access the memory it
allocates---tasks do not have access to any shared memory.
In addition, ring protection could be used to intercept system
calls performed by
a task and only permit those corresponding to our IFC language (such as
\ensuremath{\mathbf{getLabel}} or \ensuremath{\mathbf{send}}).
Dune's hardware protection mechanism would allow us to provide a concrete
implementation that is efficient and relatively simple to reason
about, but other sandboxing mechanisms could be used in place of Dune.

In this setting, the combined language of Section~\ref{sec:retrofit}
can be interpreted in the following way: calling from the target
language to the IFC language corresponds to invoking a system call.
Creating a new task with the \ensuremath{\mathbf{sandbox}} system call corresponds to
\emph{forking} a process.  Using page tables, we can ensure that
there will be no shared memory
(effectively
defining \ensuremath{\kappa(\tar{\Sigma})\mathrel{=}\tar{\Sigma}_{0}}, where \ensuremath{\tar{\Sigma}_{0}} is the set of pages necessary to bootstrap a
lightweight process).
Similarly, control over page tables and protection bits allows us to
define a \ensuremath{\mathbf{send}} system call that copies pages to our
(trusted) runtime queue; and, correspondingly, a \ensuremath{\mathbf{recv}} that copies
the pages from the runtime queue to the (untrusted) receiver.
Since C is not memory safe, conditions on these system calls are
meaningless.
We leave the implementation of this IFC system for C as future work.

\section{Extensions and Limitations}
\label{sec:extensions}
\label{sec:extensions:labeled}

While the IFC language presented thus far provides the basic
information flow primitives, actual IFC implementations
may wish to extend the minimal system with more specialized
constructs.
For example, COWL provides a labeled version of the XMLHttpRequest (XHR) object,
which is used to make network requests.
%
Our system can be extended with constructs such as labeled values,
labeled mutable references, clearance, and privileges.
For space reasons, we provide
details of this, including the soundness
proof with the extensions, in \appendixextfirst{}.
Here, we instead discuss a limitation of our formalism: the lack of
external effects.

Specifically, our embedding assumes that the target language does not
have any primitives that can induce external effects.
As discussed in Section~\ref{sec:real}, imposing this restriction
can be challenging.
Yet, external effects are crucial when implementing more complex
real-world applications.
For example, code in an IFC browser must load resources or
perform XHR to be useful.

Like labeled references, features with external effects must be
modeled in the IFC language; we must reason about the precise security
implications of features that otherwise inherently leak data.
Previous approaches have modeled external effects by internalizing the
effects as operations on labeled channels/references~\cite{stefan:addressing-covert}.
Alternatively, it is possible to model such effects as messages to/from
certain labeled tasks, an approach taken by our Node.js
implementation.
These ``special'' tasks are trusted with access to the unlabeled
primitives that can be used to perform the external effects; since the
interface to these tasks is already part of the IFC language, the
proof only requires showing that this task does not leak information.
Instead of restricting or wrapping unsafe primitives,
COWL allow for
controlled network communication at the context boundary.
(By restricting the default XHR object, for example, COWL allows code
to communicate with hosts according to the task's current label.)
\section{Related Work}
\label{sec:related}


Our information flow control system is closely related
to the coarse-grained information systems used in operating systems such
as Asbestos~\cite{efstathopoulos:asbestos},
HiStar~\cite{Zeldovich:2006}, and Flume~\cite{krohn:flume}, as well as language-based
\emph{floating-label IFC systems} such as LIO~\cite{lio},
and Breeze~\cite{Hritcu:2013:YIB:2497621.2498098}, where there is a
monotonically increased label
associated with threads of execution.
Our treatment of termination-sensitive and termination-insensitive interference
originates from Smith and Volpano~\cite{Smith:Volpano:MultiThreaded,Volpano:1997:ECF:794197.795081}.

One information flow control technique designed to handle legacy code is
secure multi-execution (SME)~\cite{Devriese:2010,Rafnson:2013}. SME runs
multiple copies of the program, one per security level, where the semantics of
I/O interactions is altered.
Bielova et
al.~\cite{Biel-etal-11-TR} use a transition system to describe SME, where the
details of the underlying language are hidden.  Zanarini et
al.~\cite{ZanariniJR13} propose a novel semantics for programs based on
interaction trees~\cite{jacobs-tutorial}, which treats programs as black-boxes
about which nothing is known, except what can be inferred from their interaction
with the environment. Similar to SME, our approach mediates I/O
operations; however, our approach only runs the program once.

One of the primary motivations behind this paper is the application of
information flow control to JavaScript.  Previous systems 
retrofitted JavaScript with fine-grained
\ifextended
IFC~\cite{Hedin:2012,ConDOM,JSFlow}.
\else
IFC~\cite{Hedin:2012,JSFlow}.
\fi
While fine-grained IFC can result
in fewer false alarms and target legacy code, it comes at the
cost of complexity: the system must accommodate the entirety of JavaScript's
semantics~\cite{JSFlow}.
By contrast, coarse-grained
approaches to security
tend to have simpler implications~\cite{Yip:2009:PBS,DeGroef:2012}.


The constructs in our IFC language, as well as the behavior of
inter-task communication, are reminiscent of distributed systems like
Erlang~\cite{Armstrong03makingreliable}.
In distributed systems, isolation is required due to
physical constraints; in information flow control, isolation is
required to enforce non-interference.  Papagiannis et
al.~\cite{Papagiannis_enforcinguser} built an information flow control
system on top of Erlang that shares some similarities to ours.
However, they do not take a floating-label approach (processes can find
out when sending a message failed due to a forbidden information
flow), nor do they provide security proofs.

There is limited work on general techniques for retrofitting
arbitrary languages with information flow control. However, one time-honored
technique is to define a fundamental calculus for which other languages can be
desugared into.  Abadi et al.~\cite{abadi+:core} motivate their core calculus of
dependency by showing how various previous systems can be encoded in it. Tse and
Zdancewic~\cite{Tse:Zdancewic:ICFP04}, in turn, show how this calculus can be
encoded in System F via parametricity.  Broberg and Sands~\cite{Broberg:2010}
encode several IFC systems into Paralocks.  However, this line of work is
primarily focused on static enforcements.
\section{Conclusion}
\label{sec:conclusion}

In this paper, we argued that when designing a coarse-grained IFC
system, it is better to start with a fully isolated, multi-task system
and work one's way back to the model of a single language equipped
with IFC.  We showed
how systems designed this way can be proved non-interferent
without needing to rely on details of the target language, and
we provided conditions on how to securely refine our formal semantics to
consider optimizations required in practice.  We connected our semantics
to two IFC implementations for JavaScript based on this formalism,
explained how our methodology improved an exiting IFC system for Haskell,
and proposed
an IFC system for C using hardware isolation.
By systematically applying ideas from IFC in operating systems to
programming languages for which isolation can be achieved,
we hope to have elucidated some of the core design principles of coarse-grained,
dynamic IFC systems.


{
\small
\paragraph{Acknowledgements}
We thank the POST 2015 anonymous reviewers,
Adriaan Larmuseau,
Sergio Maffeis, and
David Mazi\`eres
for useful comments and suggestions.
This work was funded by DARPA CRASH under contract \#N66001-10-2-4088,
by the NSF, by the AFOSR,
by multiple gifts from Google, by a gift from Mozilla,
and by the Swedish research agencies VR and the Barbro Oshers Pro
Suecia Foundation.
Deian Stefan and Edward Z. Yang were supported by the DoD through the
NDSEG.
}

{
\ifextended%
\else%
\frenchspacing\scriptsize
\setlength{\bibsep}{2pt}
\fi%
  \bibliographystyle{abbrvnat}
  \bibliography{local}
}

\ifextended
\clearpage
\balance
\appendix

\section{Full Semantics for \ensuremath{\Red{\lambda_{\text{ES}}}}}
\label{sec:app:semantics}

In Fig.~\ref{fig:ml-full} we give the full semantics for \ensuremath{\Red{\lambda_{\text{ES}}}}.  A subset
of them has been given in Fig.~\ref{fig:ml} earlier in the paper.

\begin{figure}
\begin{hscode}\SaveRestoreHook
\column{B}{@{}>{\hspre}l<{\hspost}@{}}%
\column{6}{@{}>{\hspre}l<{\hspost}@{}}%
\column{22}{@{}>{\hspre}l<{\hspost}@{}}%
\column{40}{@{}>{\hspre}l<{\hspost}@{}}%
\column{47}{@{}>{\hspre}l<{\hspost}@{}}%
\column{E}{@{}>{\hspre}l<{\hspost}@{}}%
\>[B]{}\tar{v}{}\<[6]%
\>[6]{}\Coloneqq\lambda \tar{x}.\tar{e}\mid \mathbf{true}\mid \mathbf{false}\mid \tar{a}{}\<[E]%
\\
\>[B]{}\tar{e}{}\<[6]%
\>[6]{}\Coloneqq\tar{v}\mid \tar{x}\mid \tar{e}\;\tar{e}\mid \mathbf{if}\;\tar{e}\;\mathbf{then}\;\tar{e}\;\mathbf{else}\;\tar{e}\mid \mathbf{ref}\;\tar{e}\mid \mathbin{!}\tar{e}\mid \tar{e}\mathbin{:=}\tar{e}\mid \mathbf{fix}\;\tar{e}{}\<[E]%
\\
\>[B]{}\tar{E}{}\<[6]%
\>[6]{}\Coloneqq\tar{[\cdot ]_T}\mid \tar{E}\;\tar{e}\mid \tar{v}\;\tar{E}\mid \mathbf{if}\;\tar{E}\;\mathbf{then}\;\tar{e}\;\mathbf{else}\;\tar{e}\mid \mathbf{ref}\;\tar{E}\mid \mathbin{!}\tar{E}\mid \tar{E}\mathbin{:=}\tar{e}\mid \tar{v}\mathbin{:=}\tar{E}\mid \mathbf{fix}\;\tar{E}{}\<[E]%
\\
\>[B]{}\tar{e}_{1};\tar{e}_{2}{}\<[22]%
\>[22]{}\triangleq{}(\lambda \tar{x}.\tar{e}_{2})\;\tar{e}_{1}\;{}\<[40]%
\>[40]{}\mathbf{where}\;{}\<[47]%
\>[47]{}\tar{x}\;\not\in\;\mathcal{FV}\;(\tar{e}_{2}){}\<[E]%
\\
\>[B]{}\mathbf{let}\;\tar{x}\mathrel{=}\tar{e}_{1}\;\mathbf{in}\;\tar{e}_{2}{}\<[22]%
\>[22]{}\triangleq{}(\lambda \tar{x}.\tar{e}_{2})\;\tar{e}_{1}{}\<[E]%
\ColumnHook
\end{hscode}\resethooks
\begin{mathpar}

\inferrule[T-app]
{ } {\ensuremath{\mathcal{E}_{\tar{\Sigma}}\left[(\lambda \Varid{x}.\tar{e})\;\tar{v}\right]\rightarrow\mathcal{E}_{\tar{\Sigma}}\left[\{\mskip1.5mu \tar{v}\mathbin{/}\Varid{x}\mskip1.5mu\}\;\tar{e}\right]}}

\and
\inferrule[T-ifTrue]
{ } {\ensuremath{\mathcal{E}_{\tar{\Sigma}}\left[\;\mathbf{if}\;\mathbf{true}\;\mathbf{then}\;\tar{e}_{1}\;\mathbf{else}\;\tar{e}_{2}\right]\rightarrow\mathcal{E}_{\tar{\Sigma}}\left[\tar{e}_{1}\right]}}

\and
\inferrule[T-ifFalse]
{ } {\ensuremath{\mathcal{E}_{\tar{\Sigma}}\left[\;\mathbf{if}\;\mathbf{false}\;\mathbf{then}\;\tar{e}_{1}\;\mathbf{else}\;\tar{e}_{2}\right]\rightarrow\mathcal{E}_{\tar{\Sigma}}\left[\tar{e}_{2}\right]}}

\and
\inferrule[T-ref]
{ \ensuremath{\textrm{fresh}(\tar{a})} }
{\ensuremath{\mathcal{E}_{\tar{\Sigma}}\left[\mathbf{ref}\;\tar{v}\right]\rightarrow\mathcal{E}_{\tar{\Sigma}\left[\tar{a}\mapsto{}\tar{v}\right]}\left[\tar{a}\right]}}

\and
\inferrule[T-deref]
{ \ensuremath{(\tar{a},\tar{v})\in\tar{\Sigma}} }
{\ensuremath{\mathcal{E}_{\tar{\Sigma}}\left[\mathbin{!}\tar{a}\right]\rightarrow\mathcal{E}_{\tar{\Sigma}}\left[\tar{v}\right]}}

\and
\inferrule[T-ass]
{ }
{\ensuremath{\mathcal{E}_{\tar{\Sigma}}\left[\tar{a}\mathbin{:=}\tar{v}\right]\rightarrow\mathcal{E}_{\tar{\Sigma}\left[\tar{a}\mapsto{}\tar{v}\right]}\left[\tar{v}\right]}}

\and
\inferrule[T-fix]
{ }
{\ensuremath{\mathcal{E}_{\tar{\Sigma}}\left[\mathbf{fix}\;(\lambda \Varid{x}.\Varid{e})\right]\rightarrow\mathcal{E}_{\tar{\Sigma}}\left[\{\mskip1.5mu \mathbf{fix}\;(\lambda \Varid{x}.\Varid{e})\mathbin{/}\Varid{x}\mskip1.5mu\}\;\Varid{e}\right]}}

\end{mathpar}

\caption{\ensuremath{\Red{\lambda_{\text{ES}}}}: simple untyped lambda calculus extended with booleans,
mutable references and general recursion. \ensuremath{\mathcal{FV}\;(\tar{e})} returns the set of free
variables in expression \ensuremath{\tar{e}}.}
\label{fig:ml-full}
\end{figure}

\section{Example IFC Language with a Single Heap}
\label{sec:single-heap}

As a concrete instantiation of this proof technique, we show
how to make implement our IFC language using a single heap and
ensure its non-interference using the techniques presented.
First, we can construct the restricted language
\ensuremath{L_\text{IFC}^{\mathcal{P}_\text{norefs}}(\alpha,\Red{\lambda_{\text{ES}}})}, where \ensuremath{\mathcal{P}_\text{norefs}} is
the family of always valid predicates, except for the ones for
\textsc{I-sandbox} and \textsc{I-send}, which we define as
$ P(\ensuremath{\ifc{\Varid{e}}}) = (\mathcal{AV}(\ensuremath{\ifc{\Varid{e}}}) = \emptyset{}) $
where $\mathcal{AV}(\ensuremath{\ifc{\Varid{e}}})$ denotes the set of address variables in \ensuremath{\ifc{\Varid{e}}}.
That is, we do not restrict any rules except for \textsc{I-sandbox}
and \textsc{I-send}.
Since $P$ only depends on \ensuremath{\ifc{\Varid{e}}}, which is part of the current
task and thus never erased w.r.t.\ the label of the first task,
this language satisfies non-interference by Theorem~\ref{thm:restricted}.

\begin{figure}
  
  \begin{mathpar}
    \inferrule[C-sandbox]
    {
      \mathcal{AV}(\ensuremath{\ifc{\Varid{e}}}) = \emptyset{}\\
      \ensuremath{\ifc{\Sigma}'\mathrel{=}\ifc{\Sigma}\left[\ifc{\Varid{i}}'\mapsto{}\mathbf{nil}\right]}\\
      \ensuremath{\ifc{\Varid{t}}_{1}\mathrel{=}\langle \ifc{E}\ifc{[}\ifc{\Varid{i}}'\ifc{]}\rangle^{\ifc{\Varid{i}}}_{\ifc{l}}}\\
      \ensuremath{\ifc{\Varid{t}}_\textrm{new}\mathrel{=}\langle \tar{_{\textrm{TI}}\lfloor}\ifc{\Varid{e}}\tar{\rfloor}\rangle^{\ifc{\Varid{i}}'}_{\ifc{l}}}\\
      \ensuremath{\textrm{fresh}(\ifc{\Varid{i}}')}
    }
    {\ensuremath{\ifc{\Sigma};\tar{\Sigma};\langle \ifc{E}\ifc{[}\mathbf{sandbox}\;\ifc{\Varid{e}}\ifc{]}_\ifc{I}\rangle^{\ifc{\Varid{i}}_{1}}_{\ifc{l}_{1}},\ldots\hookrightarrow\ifc{\Sigma}';\tar{\Sigma};\alpha_{{\tiny\mathrm{\text{sandbox}}}}(\ifc{\Varid{t}}_{1},\ldots,\ifc{\Varid{t}}_\textrm{new})}}
    \and
    \inferrule[C-send]
    {
      \mathcal{AV}(\ensuremath{\ifc{\Varid{e}}}) = \emptyset{}\\
      \ensuremath{\ifc{l}\flows{}\ifc{l}'}\\
      \ensuremath{\ifc{\Sigma}\;(\ifc{\Varid{i}}')\mathrel{=}\ifc{\Theta}}\\
      \ensuremath{\ifc{\Sigma}'\mathrel{=}\ifc{\Sigma}\left[\ifc{\Varid{i}}'\mapsto{}(\ifc{l}',\ifc{\Varid{i}},\ifc{\Varid{v}}),\ifc{\Theta}\right]}
    }
    {\ensuremath{\ifc{\Sigma};\tar{\Sigma};\langle \ifc{E}\ifc{[}\mathbf{send}\;\ifc{\Varid{i}}'\;\ifc{l}'\;\ifc{\Varid{v}}\ifc{]}_\ifc{I}\rangle^{\ifc{\Varid{i}}}_{\ifc{l}},\ldots\rightarrow\ifc{\Sigma};\tar{\Sigma};\alpha_{{\tiny\mathrm{\text{step}}}}(\langle \langle\rangle\rangle^{\ifc{\Varid{i}}}_{\ifc{l}},\ldots)}}
  \end{mathpar}
  
  \caption{A selection of the reduction rules for \ensuremath{L_\text{IFC}^{\Red{\text{Heap}}}(\alpha)}.}
  \label{fig:concrete}
\end{figure}

The essential parts of the semantics for the concrete language
with a single heap,
which we call \ensuremath{L_\text{IFC}^{\Red{\text{Heap}}}(\alpha)},
are given in Fig.~\ref{fig:concrete}.  Most rules are
straight-forward translations of the rules in Figs.~\ref{fig:ifc}
and~\ref{fig:embedding} but for a single heap.  For conciseness, we
only show the interesting ones.
Now, we can show an isomorphism between this language and
\ensuremath{L_\text{IFC}^{\mathcal{P}_\text{norefs}}(\alpha,\Red{\lambda_{\text{ES}}})}, which
(by Theorem~\ref{thm:iso-tsni} and~\ref{thm:iso-tini}) guarantees
non-interference for an appropriate scheduling policy \ensuremath{\alpha}.

To this end, we represent addresses in the concrete language as
pairs $(\ensuremath{\ifc{\Varid{i}}},\ensuremath{\tar{a}})$ where \ensuremath{\ifc{\Varid{i}}} is a task identifier, and \ensuremath{\tar{a}} an
address in the abstract system\footnote{Note that this does
  not make the isomorphism trivial, as in the single heap, there
  is nothing preventing task 1 to access an address (2,\ensuremath{\tar{a}}).
  Furthermore, it is common to represent addresses in this way
  for efficient garbage collection of dead tasks.}.
We also formulate the following well-formedness condition for
configurations:
\[
  \ensuremath{\operatorname{wf}(c)} = \forall \ensuremath{\langle \ifc{\Varid{e}}\rangle^{\ifc{\Varid{i}}}_{\ifc{l}}} \in \ensuremath{\ifc{c}}.\ 
  \{ (\ensuremath{\ifc{\Varid{i}}'},\ensuremath{\ifc{\Varid{e}}'}) \in \mathcal{AV}(\ensuremath{\ifc{\Varid{e}}}) \ \vert\ \ensuremath{\ifc{\Varid{i}}} \neq \ensuremath{\ifc{\Varid{i}}'} \} = \emptyset
\]
Essentially, every address in a given task must have the correct
identifier as the first part of the address.  It is easy to
see that the initial configuration satisfies this condition, and
any step in the concrete semantics preserves the condition.
Therefore, we only need to consider well-formed configurations,
which allows us to give the two required functions
\ensuremath{\Varid{f}} and \ensuremath{\Varid{f}^{-1}} for the isomorphism.  For conciseness, we only give
the interesting parts of their definition,
and leave out the straight-forward proof that they
actually provide an isomorphism.
\begin{itemize}
  \item Addresses can be directly translated with
  $\ensuremath{\Varid{f}}((\ensuremath{\ifc{\Varid{i}}},\ensuremath{\tar{a}}))=\ensuremath{\tar{a}}$, and $\ensuremath{\Varid{f}^{-1}}(\ensuremath{\tar{a}})=\ensuremath{(\ifc{\Varid{i}},\tar{a})}$ for
  an address \ensuremath{\tar{a}} that occurs in task \ensuremath{\ifc{\Varid{i}}}.
  \item \ensuremath{\Varid{f}} splits the single heap into multiple heaps based on
  the \ensuremath{\ifc{\Varid{i}}} of the addresses.  \ensuremath{\Varid{f}^{-1}} produces a single heap
  by translating the addresses and collapsing everything to a single
  store.
\end{itemize}

\section{Extending the Core Calculus}
\label{sec:appendix-extensions}

As mentioned in the main body of this paper,
actual IFC implementations
may wish to extend the minimal system with more specialized
constructs.
In this section we show how to extend the language with several such
constructs.
%

\subsection{Labeled values}
In traditional language-based dynamic IFC systems, a label is
associated with values.
Hence, a program that, for example, simply writes labeled messages to
a labeled log can operate on both public and sensitive values.
Similarly, a task that receives a sensitive value and forwards it
to another task does not have be be at a sensitive level, if the
value is not inspected.
In its simplest form, our coarse grained system requires that the
current label of a task be at least at the level of the sensitive data
to reflect the fact that such data is in scope.

If such fine-grained labeling of values is required, our base IFC
system can be extended with explicitly labeled
values, much like those of LIO and
Breeze~\cite{lio, Hritcu:2013:YIB:2497621.2498098}: \ensuremath{\ifc{\Varid{v}}\Coloneqq\cdots\;|\;\mathbf{Labeled}\;\ifc{l}\;\ifc{\Varid{e}}}.
Following LIO, we say that the expression \ensuremath{\ifc{\Varid{e}}} is protected by label \ensuremath{\ifc{l}},
while the label \ensuremath{\ifc{l}} itself is protected by the task's current label.
The label of such values can be inspected the task without
requiring the current label to be raised.
However, when a task wishes to inspect the protected value \ensuremath{\ifc{\Varid{e}}}, it
must first raise its label to at least \ensuremath{\ifc{l}} to reflect that it is
incorporating data at such sensitivity level in its scope.
When creating labeled values the label \ensuremath{\ifc{l}} must be above
the current label; otherwise it cannot be said that protection has
been transferred from the current label to \ensuremath{\ifc{l}}.

In Fig.~\ref{fig:labeled-vals}, we formally show how to add this
extension to the language.
We assume that the constructor \ensuremath{\mathbf{Labeled}} is not part
of the surface syntax, but rather an internal construct.

\begin{figure}
        \begin{hscode}\SaveRestoreHook
\column{B}{@{}>{\hspre}l<{\hspost}@{}}%
\column{5}{@{}>{\hspre}l<{\hspost}@{}}%
\column{9}{@{}>{\hspre}l<{\hspost}@{}}%
\column{E}{@{}>{\hspre}l<{\hspost}@{}}%
\>[5]{}\ifc{\Varid{v}}{}\<[9]%
\>[9]{}\Coloneqq\cdots\mid \mathbf{Labeled}\;\ifc{l}\;\ifc{\Varid{e}}{}\<[E]%
\\
\>[5]{}\ifc{\Varid{e}}{}\<[9]%
\>[9]{}\Coloneqq\cdots\mid \mathbf{label}\;\ifc{\Varid{e}}\;\ifc{\Varid{e}}\mid \mathbf{unlabel}\;\ifc{\Varid{e}}\mid \mathbf{labelOf}\;\ifc{\Varid{e}}{}\<[E]%
\\
\>[5]{}\ifc{E}{}\<[9]%
\>[9]{}\Coloneqq\cdots\mid \mathbf{label}\;\ifc{E}\;\ifc{\Varid{e}}\mid \mathbf{unlabel}\;\ifc{E}\mid \mathbf{labelOf}\;\ifc{E}{}\<[E]%
\ColumnHook
\end{hscode}\resethooks
  \begin{mathpar}
    \inferrule[I-label]
    {\ensuremath{\ifc{l}\;\flows{}\;\ifc{l}'}}
    {\ensuremath{\mathcal{E}_{\ifc{\Sigma}}^{\ifc{\Varid{i}},\ifc{l}}\left[\mathbf{label}\;\ifc{l}'\;\ifc{\Varid{e}}\right]\rightarrow\mathcal{E}_{\ifc{\Sigma}}^{\ifc{\Varid{i}},\ifc{l}}\left[\mathbf{Labeled}\;\ifc{l}'\;\ifc{\Varid{e}}\right]}}
    \and
    \inferrule[I-unlabel]
    {}
    {\ensuremath{\mathcal{E}_{\ifc{\Sigma}}^{\ifc{\Varid{i}},\ifc{l}}\left[\mathbf{unlabel}\;(\mathbf{Labeled}\;\ifc{l}'\;\ifc{\Varid{e}})\right]\rightarrow\mathcal{E}_{\ifc{\Sigma}}^{\ifc{\Varid{i}},\ifc{l}\;\lub\;\ifc{l}'}\left[\ifc{\Varid{e}}\right]}}
    \and
    \inferrule[I-labelOf]
    {\ensuremath{}}
    {\ensuremath{\mathcal{E}_{\ifc{\Sigma}}^{\ifc{\Varid{i}},\ifc{l}}\left[\mathbf{labelOf}\;(\mathbf{Labeled}\;\ifc{l}'\;\ifc{\Varid{e}})\right]\rightarrow\mathcal{E}_{\ifc{\Sigma}}^{\ifc{\Varid{i}},\ifc{l}}\left[\ifc{l}'\right]}}
  \end{mathpar}
  \caption{Syntax and semantics for labeled values.  These rules are
    understood to be an addition to the existing rules given earlier.}
  \label{fig:labeled-vals}
  \end{figure}

\subsection{Labeled mutable references/variables/channels}
Extending the calculus with other labeled features, such as
references, mutable variables (MVars)~\cite{CH96}, or channels,
can be done in a similar manner: these references are implemented
in the IFC language, separately from any preexisting notions of
mutable references in the target language.
%
%
There is some minor additional state to track: specifically, by amending \ensuremath{\ifc{\Sigma}}, as in~\cite{lio,
stefan:addressing-covert}, we can allow threads to use these
constructs to synchronize, or communicate with constructs other than
\ensuremath{\mathbf{send}}/\ensuremath{\mathbf{recv}} in a safe manner.
For example, when extending the calculus with labeled references, \ensuremath{\ifc{\Sigma}}
additionally contains a store that maps addresses to a value
and a label
which can be read and written to by different tasks through a labeled
reference implementations.

In Fig.~\ref{fig:labeled-refs} details labeled references formally.
The construct \ensuremath{\ifc{\Varid{a}}_{\ifc{l}}} is internal in the labeled reference implementation,
and not part of the surface syntax.
The changes to the language for labeled values and
references require us to update the erasure function
\ensuremath{\varepsilon_{\ifc{l}}}, whose full definition is shown in Fig.~\ref{fig:erasure2}.

\begin{figure}
        \begin{hscode}\SaveRestoreHook
\column{B}{@{}>{\hspre}l<{\hspost}@{}}%
\column{5}{@{}>{\hspre}l<{\hspost}@{}}%
\column{9}{@{}>{\hspre}l<{\hspost}@{}}%
\column{14}{@{}>{\hspre}l<{\hspost}@{}}%
\column{19}{@{}>{\hspre}l<{\hspost}@{}}%
\column{E}{@{}>{\hspre}l<{\hspost}@{}}%
\>[5]{}\ifc{\Varid{v}}{}\<[9]%
\>[9]{}\Coloneqq\cdots\mid \ifc{\Varid{a}}_{\ifc{l}}{}\<[E]%
\\
\>[5]{}\ifc{\Varid{e}}{}\<[9]%
\>[9]{}\Coloneqq\cdots\mid \mathbf{new}\;\ifc{\Varid{e}}\;\ifc{\Varid{e}}\mid \mathbf{read}\;\ifc{\Varid{e}}\mid \mathbf{write}\;\ifc{\Varid{e}}\;\ifc{\Varid{e}}{}\<[E]%
\\
\>[5]{}\ifc{E}{}\<[9]%
\>[9]{}\Coloneqq\cdots\mid \mathbf{new}\;\ifc{E}\;\ifc{\Varid{e}}\mid \mathbf{new}\;\ifc{l}\;\ifc{E}\mid \mathbf{read}\;E{}\<[E]%
\\
\>[9]{}\hsindent{10}{}\<[19]%
\>[19]{}\mid \mathbf{write}\;\ifc{E}\;\ifc{\Varid{e}}\mid \mathbf{write}\;\ifc{\Varid{a}}_{\ifc{l}}\;\ifc{E}{}\<[E]%
\\
\>[5]{}\ifc{\Sigma}{}\<[9]%
\>[9]{}\Coloneqq{}\<[14]%
\>[14]{}\cdots\mid \ifc{\Sigma}\left[\ifc{\Varid{a}}_{\ifc{l}}\mapsto{}\ifc{\Varid{v}}\right]{}\<[E]%
\ColumnHook
\end{hscode}\resethooks
  
  \begin{mathpar}
    \inferrule[I-new]
    {
      \ensuremath{\ifc{l}\;\flows{}\;\ifc{l}'} \\
      \ensuremath{\textrm{fresh}(\ifc{\Varid{a}})}\\
      \ensuremath{\ifc{\Sigma}'\mathrel{=}\ifc{\Sigma}\left[\ifc{\Varid{a}}_{\ifc{l}'}\mapsto{}\ifc{\Varid{v}}\right]}
    }
    {\ensuremath{\mathcal{E}_{\ifc{\Sigma}}^{\ifc{\Varid{i}},\ifc{l}}\left[\mathbf{new}\;\ifc{l}'\;\ifc{\Varid{v}}\right]\rightarrow\mathcal{E}_{\ifc{\Sigma}'}^{\ifc{\Varid{i}},\ifc{l}}\left[\ifc{\Varid{a}}_{\ifc{l}'}\right]}}
    \and
    \inferrule[I-read]
    {}
    {\ensuremath{\mathcal{E}_{\ifc{\Sigma}}^{\ifc{\Varid{i}},\ifc{l}}\left[\mathbf{read}\;\ifc{\Varid{a}}_{\ifc{l}'}\right]\rightarrow\mathcal{E}_{\ifc{\Sigma}}^{\ifc{\Varid{i}},\ifc{l}\;\lub\;\ifc{l}'}\left[\ifc{\Sigma}(\ifc{\Varid{a}}_{\ifc{l}'})\right]}}
    \and
    \inferrule[I-write]
    {
      \ensuremath{\ifc{l}\;\flows{}\;\ifc{l}'}\\
      \ensuremath{\ifc{\Sigma}'\mathrel{=}\ifc{\Sigma}\left[\ifc{\Varid{a}}_{\ifc{l}'}\mapsto{}\ifc{\Varid{v}}\right]}
    }
    {\ensuremath{\mathcal{E}_{\ifc{\Sigma}}^{\ifc{\Varid{i}},\ifc{l}}\left[\mathbf{write}\;\ifc{\Varid{a}}_{\ifc{l}'}\;\ifc{\Varid{v}}\right]\rightarrow\mathcal{E}_{\ifc{\Sigma}'}^{\ifc{\Varid{i}},\ifc{l}}\left[\langle\rangle\right]}}
    \and
    \inferrule[I-labelOf2]
    {}
    {\ensuremath{\mathcal{E}_{\ifc{\Sigma}}^{\ifc{\Varid{i}},\ifc{l}}\left[\mathbf{labelOf}\;\ifc{\Varid{a}}_{\ifc{l}'}\right]\rightarrow\mathcal{E}_{\ifc{\Sigma}}^{\ifc{\Varid{i}},\ifc{l}}\left[\ifc{l}'\right]}}
  \end{mathpar}
  \caption{Syntax and semantics for labeled references.  These rules are
    understood to be an addition to the existing rules given earlier.}
  \label{fig:labeled-refs}
  \end{figure}

\begin{figure}
  \begin{align*}
  &\ensuremath{\varepsilon_{\ifc{l}}(\ifc{\Sigma};\ifc{\Varid{ts}})\mathrel{=}\varepsilon_{\ifc{l}}(\ifc{\Sigma});\text{filter}\;(\lambda \ifc{\Varid{t}}.\ifc{\Varid{t}}\mathrel{=}\bullet)\;(\text{map}\;\varepsilon_{\ifc{l}}\;\ifc{\Varid{ts}})}\\
  &\ensuremath{\langle \tar{\Sigma}, \ifc{\Varid{e}}\rangle^{\ifc{\Varid{i}}}_{\ifc{l}'}} \begin{cases}
  \ensuremath{\bullet} & \ensuremath{\ifc{l}'\;\not\flows{}\;\ifc{l}} \\
  \ensuremath{\langle \varepsilon_{\ifc{l}}(\tar{\Sigma}), \varepsilon_{\ifc{l}}(\ifc{\Varid{e}})\rangle^{\ifc{\Varid{i}}}_{\ifc{l}'}} & \text{otherwise}
  \end{cases} \\
  &\ensuremath{\varepsilon_{\ifc{l}}(\mathbf{Labeled}\;\ifc{l}'\;\ifc{\Varid{e}})}= \begin{cases}
  \ensuremath{\mathbf{Labeled}\;\ifc{l}'\;\bullet} & \ensuremath{\ifc{l}'\;\not\flows{}\;\ifc{l}} \\
  \ensuremath{\mathbf{Labeled}\;\ifc{l}'\;\ifc{\Varid{e}}} & \text{otherwise}
  \end{cases} \\
  &\ensuremath{\varepsilon_{\ifc{l}}(\emptyset)\mathrel{=}\emptyset}\\
  &\ensuremath{\varepsilon_{\ifc{l}}(\ifc{\Sigma}\left[\ifc{\Varid{i}}\mapsto{}\ifc{\Theta}\right])\mathrel{=}} \begin{cases}
  \ensuremath{\varepsilon_{\ifc{l}}(\ifc{\Sigma})} & \text{\ensuremath{\ifc{l}'\;\not\flows{}\;\ifc{l}}, where \ensuremath{\ifc{l}'} is the label of thread \ensuremath{\ifc{\Varid{i}}}}\\
  \ensuremath{\varepsilon_{\ifc{l}}(\ifc{\Sigma})\left[\ifc{\Varid{i}}\mapsto{}\varepsilon_{\ifc{l}}(\ifc{\Theta})\right]} & \text{otherwise}
  \end{cases} \\
  &\ensuremath{\varepsilon_{\ifc{l}}(\ifc{\Sigma}\left[\ifc{\Varid{a}}_{\ifc{l}'}\mapsto{}\ifc{\Varid{v}}\right])\mathrel{=}} \begin{cases}
  \ensuremath{\varepsilon_{\ifc{l}}(\ifc{\Sigma})\left[\ifc{\Varid{a}}_{\ifc{l}'}\mapsto{}\bullet\right]} & \text{\ensuremath{\ifc{l}'\;\not\flows{}\;\ifc{l}}}\\
  \ensuremath{\varepsilon_{\ifc{l}}(\ifc{\Sigma})\left[\ifc{\Varid{a}}_{\ifc{l}'}\mapsto{}\varepsilon_{\ifc{l}}(\ifc{\Varid{v}})\right]} & \text{otherwise}
  \end{cases} \\
  &\ensuremath{\varepsilon_{\ifc{l}}(\ifc{\Theta})\mathrel{=}\ifc{\Theta}\preceq\ifc{l}}
  \end{align*}
  \caption{Erasure function for the full IFC language, with all extensions.
    In all cases that are not specified, including target-language constructs,
    \ensuremath{\varepsilon_{\ifc{l}}} is applied homomorphically
    (e.g., \ensuremath{\varepsilon_{\ifc{l}}(\mathbf{setLabel}\;\Varid{e})\mathrel{=}\mathbf{setLabel}\;\varepsilon_{\ifc{l}}(\Varid{e})}).
    This definition replaces the one from Fig.~\ref{fig:erasure}, which
    is for the IFC language without extensions.}
  \label{fig:erasure2}
\end{figure}

\subsection{Clearance}
\label{sec:clearance}
Systems like LIO, COWL, and Breeze additionally provide a discretionary access
control (DAC) mechanism---called \emph{clearance}---at the language
level~\cite{Hritcu:2013:YIB:2497621.2498098, lio}.
This mechanisms is used to restrict a computation from allocating and
accessing data (or communicating with entities) above a specified
label, the clearance.
Amending our IFC language with clearance is straight forward,
and, can be done using our notation of a restricted language.
To this end, we first extend tasks to track a clearance label
alongside the current label, and amend the core IFC language with two
new terminals for retrieving and setting this value.
Since this extension only adds a per-task mutable variable whose value
has no influence on the system, all security guarantees still
hold, by essentially the same proofs.
However, this does not implement any DAC mechanism yet.
To do so, we can restrict the language with a family of predicates
$\mathcal{P}_\text{clearance}$:
All rules that
raise the current label (e.g., \textsc{I-setLabel}), perform
allocation (e.g., \textsc{I-sandbox} and \textsf{I-send}), or set the
clearance (clearance should not be arbitrarily raised), a predicate
that uses the clearance to impose DAC is used.
For instance, the predicate for \textsc{I-setLabel} prevents the
current label from being raised above the clearance (and thus permit
reads above the clearance).  The predicate $P := \ensuremath{\ifc{l}\;\flows{}\;\ifc{l}'}$ achieves this restriction, where \ensuremath{\ifc{l}'} is the
clearance and \ensuremath{\ifc{l}} is the current
label.
The other predicates are defined in a similar way and omitted for
brevity.

\subsection{Privileges}
Decentralized IFC extends IFC with the decentralized label model of
Myers and Liskov~\cite{myers:dlm} to allow for more general
applications, including systems consisting of mutually distrustful
parties.  In a decentralized system, a computation is executed with a
set of \emph{privileges}, which, when exercised, allow the computation
to declassify data (e.g., by lowering the current label).
Practical IFC systems
(e.g.,~\cite{Zeldovich:2006, lio,
  Hritcu:2013:YIB:2497621.2498098, myers:jif}) rely on privileges to
implement many applications.
%
%
The challenge with such an extension lies in the precise
security guarantees that must be proved, which to the best of our
knowledge is an open research problem.

Our implementation for Node.js and COWL both provide privileges, but
we have not formalized this part any further.

\section{Non-Interference Proof}
\label{sec:appendix}
\label{sec:app:proof}

In this section we prove the theorems we have stated in the paper.
Note that we prove soundness of the system including the formally
defined extensions from Appendix~\ref{sec:appendix-extensions}.
We first observe that the non-interference claims for the languages
\ensuremath{L_\text{IFC}(\textsc{Seq},\Red{\lambda})} and \ensuremath{L_\text{IFC}(\textsc{RR},\Red{\lambda})}
in Theorems~\ref{thm:rr-tsni} and~\ref{thm:seq-tini} follow directly
from Theorem~\ref{thm:restricted}, where the set
of predicates is the set of always valid predicates (i.e., no restriction).

Before we proceed with the proof of Theorem~\ref{thm:restricted},
we state and proof two lemmas we will use.

\begin{lemma}
  \label{lemma:high-not-blocking}
  For any task \ensuremath{\ifc{\Varid{t}}}, task lists \ensuremath{\ifc{\Varid{ts}}}, store \ensuremath{\ifc{\Sigma}}, and label \ensuremath{\ifc{l}}, if
  $\ensuremath{\varepsilon_{\ifc{l}}(\ifc{\Varid{t}})}=\ensuremath{\bullet}$, then there exists a task list
  \ensuremath{\ifc{\Varid{ts}}'} and a store \ensuremath{\ifc{\Sigma}'} such that
  \begin{align}
  \ensuremath{\ifc{\Sigma};\ifc{\Varid{t}},\ifc{\Varid{ts}}} \ensuremath{\hookrightarrow} \ensuremath{\ifc{\Sigma}';\ifc{\Varid{ts}},\ifc{\Varid{ts}}'} \label{eq:hnb-1} \\
  \ensuremath{\varepsilon_{\ifc{l}}(\ifc{\Varid{ts}}')}=\ensuremath{\mathbf{nil}} \label{eq:hnb-2}\\
  \ensuremath{\varepsilon_{\ifc{l}}(\ifc{\Sigma}')}=\ensuremath{\varepsilon_{\ifc{l}}(\ifc{\Sigma})} \label{eq:hnb-3}
  \end{align}
\end{lemma}
\begin{proof}
  From $\ensuremath{\varepsilon_{\ifc{l}}(\ifc{\Varid{t}})}=\ensuremath{\bullet}$ we know that the current label \ensuremath{\lcurr} of \ensuremath{\ifc{\Varid{t}}}
  must be above \ensuremath{\ifc{l}}.  Furthermore, tasks can always take a step (if no
  regular rule applies, then \textsc{I-noStep} can be used), and thus
  we consider all rules that could be applied to execute \ensuremath{\ifc{\Varid{t}}}.
  \begin{description}
    \item[Case \textsc{I-noStep} and \textsc{I-done}]
    In this case, the task \ensuremath{\ifc{\Varid{t}}} is dropped,
    and thus \ensuremath{\ifc{\Varid{ts}}'\mathrel{=}\mathbf{nil}} and \ensuremath{\ifc{\Sigma}'\mathrel{=}\ifc{\Sigma}} satisfy
    conditions~\eqref{eq:hnb-2} and~\eqref{eq:hnb-3}.
    \item[Case \textsc{I-sandbox}]
    The newly created task has a label of at least \ensuremath{\lcurr}, and will thus be
    erased, as required by condition~\eqref{eq:hnb-2}.  Furthermore, the
    state only changes for the newly created thread, and thus the state
    change is erased, showing~\eqref{eq:hnb-3}.
  \end{description}
  In all other rules, no new tasks are created, and thus \ensuremath{\ifc{\Varid{ts}}'} consists of just
  the one task \ensuremath{\ifc{\Varid{t}}'}, to which \ensuremath{\ifc{\Varid{t}}} executed.  Since the tasks label can
  only increase, \ensuremath{\ifc{\Varid{t}}'} is still erased, showing condition~\eqref{eq:hnb-2}.
  We are left to show condition~\eqref{eq:hnb-3} for the remaining rules.
  \begin{description}
    \item[Case \textsc{I-send}]
    A new message triple with label \ensuremath{\ifc{l}'} gets added to the message
    queue of the receiving thread.  However, since \ensuremath{\lcurr\;\flows{}\;\ifc{l}'},
    the triple will get erased.
    \item[Case \textsc{I-recv} and \textsc{I-noRecv}]
    In this case, only the queue of
    task \ensuremath{\ifc{\Varid{t}}} can change, which gets erased.
    \item[Case \textsc{I-new}] The newly allocated address has to be at a
    label at least as high as \ensuremath{\lcurr}, and will thus be erased.
    \item[Case \textsc{I-write}] Only addresses with a label \ensuremath{\ifc{l}'} above
    \ensuremath{\lcurr} can be written, thus the change in \ensuremath{\ifc{\Sigma}_{1}} will get erased.
    \item[Otherwise.]  None of the other rules modify the state \ensuremath{\ifc{\Sigma}}, and
    thus \ensuremath{\ifc{\Sigma}'\mathrel{=}\ifc{\Sigma}} will trivially satisfy condition~\eqref{eq:hnb-3}.
  \end{description}
  \qed
\end{proof}

\begin{lemma}
  \label{lemma:rr-tsni-general}
  We consider, for any target language \ensuremath{\Red{\lambda}},
  the restricted IFC language \ensuremath{L_\text{IFC}^{\mathcal{P}}(\alpha,\Red{\lambda})}
  (according to Definition~\ref{def:restricted}).
  Then,
  for any configurations \ensuremath{\ifc{c}_{1}}, \ensuremath{\ifc{c}_{1}'}, \ensuremath{\ifc{c}_{2}}, and label \ensuremath{\ifc{l}} where
  \begin{equation} \label{eq:tsni-lemma-lhs}
  \ensuremath{\ifc{c}_{1}} \approx_{\ensuremath{\ifc{l}}} \ensuremath{\ifc{c}_{2}}
  \qquad \text{and} \qquad
  \ensuremath{\ifc{c}_{1}} \ensuremath{\hookrightarrow} \ensuremath{\ifc{c}_{1}'}
  \end{equation}
  there exists a configuration \ensuremath{\ifc{c}_{2}'} such that
  \begin{equation} \label{eq:tsni-lemma-rhs}
  \ensuremath{\ifc{c}_{1}'} \approx_{\ensuremath{\ifc{l}}} \ensuremath{\ifc{c}_{2}'}
  \qquad \text{and} \qquad
  \ensuremath{\ifc{c}_{2}} \ensuremath{\hookrightarrow}^* \ensuremath{\ifc{c}_{2}'}
  \ \text{.}
  \end{equation}
\end{lemma}
\begin{proof}
  First, we observe there must be at least one task in \ensuremath{\ifc{c}_{1}}, otherwise
  it could not take a step.  Thus, \ensuremath{\ifc{c}_{1}} is of the form
  \ensuremath{\ifc{\Sigma}_{1};\ifc{\Varid{t}}_{1},\ifc{\Varid{ts}}_{1}}.
  Furthermore, let \ensuremath{\ifc{c}_{2}} be \ensuremath{\ifc{\Sigma}_{2};\ifc{\Varid{ts}}_{2}}.
  Consider two cases:
  \begin{itemize}
    \item $\ensuremath{\varepsilon_{\ifc{l}}(\ifc{\Varid{t}}_{1})}=\ensuremath{\bullet}$.
    By the definition of \ensuremath{\varepsilon_{\ifc{l}}}, we know that \ensuremath{\ifc{l}\;\flows{}\;\lcurr}
    where \ensuremath{\lcurr} is the label of \ensuremath{\ifc{\Varid{t}}_{1}}.
    In this case, we do not need to take a step for
    \ensuremath{\ifc{c}_{2}}, because \ensuremath{\ifc{c}_{2}'\mathrel{=}\ifc{c}_{2}} will already be \ensuremath{\ifc{l}}-equivalent to \ensuremath{\ifc{c}_{1}'}.
    To show this, note that the tasks \ensuremath{\ifc{\Varid{ts}}_{1}} in \ensuremath{\ifc{c}_{1}} are left in the
    same order and unmodified (the scheduling policy only
    modifies the first task). The task \ensuremath{\ifc{\Varid{t}}_{1}} either
    gets dropped (by \textsc{I-noStep}), or
    transforms into a task \ensuremath{\ifc{\Varid{t}}_{1}'} as well as potentially spawning a new
    task \ensuremath{\ifc{\Varid{t}}_{1}''}.  Since both \ensuremath{\ifc{\Varid{t}}_{1}'} and \ensuremath{\ifc{\Varid{t}}_{1}''} have a label that is
    at least as high as the label of \ensuremath{\ifc{\Varid{t}}_{1}} (can be seen
    by inspecting all reduction rules), they will get filtered
    by \ensuremath{\varepsilon_{\ifc{l}}} in \ensuremath{\ifc{c}_{1}'}.  Therefore, the \ensuremath{\ifc{l}}-equivalence of the
    task list is guaranteed.
    Lets consider the possible changes to \ensuremath{\ifc{\Sigma}_{1}}:
    Only five reduction interact with \ensuremath{\ifc{\Sigma}_{1}},
    thus it suffices to consider these cases:
    \begin{description}
      \item[Case \textsc{I-send}]
      A new message triple with label \ensuremath{\ifc{l}'} gets added to the message
      queue of the receiving thread.  However, since \ensuremath{\lcurr\;\flows{}\;\ifc{l}'},
      the triple will get erased.
      \item[Case \textsc{I-recv} and \textsc{I-noRecv}]
      In this case, only the queue of
      task \ensuremath{\ifc{\Varid{t}}_{1}} can change, which gets erased.
      \item[Case \textsc{I-new}] The newly allocated address has to be at a
      label at least as high as \ensuremath{\lcurr}, and will thus be erased.
      \item[Case \textsc{I-write}] Only addresses with a label \ensuremath{\ifc{l}'} above
      \ensuremath{\lcurr} can be written, thus the change in \ensuremath{\ifc{\Sigma}_{1}} will get erased.
    \end{description}
    This ensures that $\ensuremath{\ifc{c}_{1}'}\approx_{\ensuremath{\ifc{l}}}\ensuremath{\ifc{c}_{2}'}$, as well as
    $\ensuremath{\ifc{c}_{2}} \ensuremath{\hookrightarrow}^* \ensuremath{\ifc{c}_{2}'}$ (in zero steps), as claimed.
    \item $\ensuremath{\varepsilon_{\ifc{l}}(\ifc{\Varid{t}}_{1})}\neq\ensuremath{\bullet}$.
    By the definition of \ensuremath{\varepsilon_{\ifc{l}}}, the task list \ensuremath{\ifc{\Varid{ts}}_{2}}
    in \ensuremath{\ifc{c}_{2}} must be of the
    form \ensuremath{\ifc{\Varid{ts}}_{2}',\ifc{\Varid{t}}_{2},\ifc{\Varid{ts}}_{2}''} (for some task lists \ensuremath{\ifc{\Varid{ts}}_{2}'}, \ensuremath{\ifc{\Varid{ts}}_{2}''} and
    some task \ensuremath{\ifc{\Varid{t}}_{2}})
    where
    \begin{align}
    \ensuremath{\varepsilon_{\ifc{l}}(\ifc{\Varid{ts}}_{2}')} = \ensuremath{\mathbf{nil}} \\
    \ensuremath{\varepsilon_{\ifc{l}}(\ifc{\Varid{t}}_{2})} = \ensuremath{\varepsilon_{\ifc{l}}(\ifc{\Varid{t}}_{1})} \\
    \ensuremath{\varepsilon_{\ifc{l}}(\ifc{\Varid{ts}}_{2}'')} = \ensuremath{\varepsilon_{\ifc{l}}(\ifc{\Varid{ts}}_{1})}
    \end{align}
    (where \ensuremath{\mathbf{nil}} is the empty list of tasks).
    Now, intuitively we will first execute a number of steps to process
    the tasks in \ensuremath{\ifc{\Varid{ts}}_{2}'} (execute them one step and move them to the back
    of the task list, or drop them if they are done or stuck).  Then, the task
    \ensuremath{\ifc{\Varid{t}}_{2}} can take the same step as \ensuremath{\ifc{\Varid{t}}_{1}}, which will result in a configuration
    \ensuremath{\ifc{c}_{2}'}
    with the desired properties.
    More formally, we can proceed as follows:
    
    First, we can apply Lemma~\ref{lemma:high-not-blocking} continuously
    for all the task in \ensuremath{\ifc{\Varid{ts}}_{2}'},
    until we reach a configuration \ensuremath{\ifc{c}_{2}''\mathrel{=}\ifc{\Sigma}_{2}';\ifc{\Varid{t}}_{2},\ifc{\Varid{ts}}_{2}'',\ifc{\Varid{ts}}_{2}'''}
    for some \ensuremath{\ifc{\Varid{ts}}_{2}'''} such
    that \ensuremath{\varepsilon_{\ifc{l}}(\ifc{\Varid{ts}}_{2}''')\mathrel{=}\mathbf{nil}} and \ensuremath{\varepsilon_{\ifc{l}}(\ifc{\Sigma}_{2})} = \ensuremath{\varepsilon_{\ifc{l}}(\ifc{\Sigma}_{2}')}.
    We note that \ensuremath{\varepsilon_{\ifc{l}}(\ifc{c}_{1})\mathrel{=}\varepsilon_{\ifc{l}}(\ifc{c}_{2}'')} (by the definition of
    \ensuremath{\varepsilon_{\ifc{l}}}).
    
    Now, the first task \ensuremath{\ifc{\Varid{t}}_{2}} in \ensuremath{\ifc{c}_{2}''} is \ensuremath{\ifc{l}}-equivalent to the task \ensuremath{\ifc{\Varid{t}}_{1}}.
    This implies that the two tasks must have the same id, label and
    can only differ in the expression or store if some subexpression
    is of the form \ensuremath{\mathbf{Labeled}\;\ifc{l}'\;\ifc{\Varid{e}}}.  In this case, the expression \ensuremath{\ifc{\Varid{e}}} could
    be different in the two threads if \ensuremath{\lcurr\;\flows{}\;\ifc{l}'}.  However, none of the reduction rules
    depend on an expression in that position, and there is never a
    hole in that position
    where evaluation could take place.  Thus, the same rules will syntactically
    match for both task, and we are left to argue that all premises
    evaluate to the same values for \ensuremath{\ifc{\Varid{t}}_{1}} and \ensuremath{\ifc{\Varid{t}}_{2}}, as well as that
    the resulting states \ensuremath{\ifc{\Sigma}_{1}'} and
    \ensuremath{\ifc{\Sigma}_{2}''} are \ensuremath{\ifc{l}}-equivalent.
    The additional premises $P$ that follow
    the condition in Definition~\ref{def:restricted} are not a problem,
    since those
    predicates only depend on \ensuremath{\varepsilon_{\ifc{l}}(\ifc{c}_{1})}, which is equivalent
    to \ensuremath{\varepsilon_{\ifc{l}}(\ifc{c}_{2}'')}, and thus those predicates evaluate in the same way.
    All other premises are either on the threads labels (which are the same),
    or on the state \ensuremath{\ifc{\Sigma}_{1}}, or \ensuremath{\ifc{\Sigma}_{2}'}, respectively.  Because
    \ensuremath{\varepsilon_{\ifc{l}}(\ifc{\Sigma}_{1})\mathrel{=}\varepsilon_{\ifc{l}}(\ifc{\Sigma}_{2}')}, all of these also evaluate in the same way,
    as can be seen by simply considering all rules that involve or
    change the state:
    \begin{description}
      \item[Case \textsc{I-send}]
      Here, the task \ensuremath{\ifc{\Varid{t}}_{2}} will send the same message to the same
      receiver queue. This
      queue is either completely erased, or it is \ensuremath{\ifc{l}}-equivalent.  In both
      cases, \ensuremath{\ifc{l}}-equivalence of \ensuremath{\ifc{\Sigma}_{1}'} and \ensuremath{\ifc{\Sigma}_{2}'} is preserved.
      \item[Case \textsc{I-recv} and \textsc{I-noRecv}]
      When the tasks are receiving a message, then by the reduction rules
      we know that they first filter the queue by the label
      \ensuremath{\lcurr} of \ensuremath{\ifc{\Varid{t}}_{1}}.  We
      also know that the queues are equivalent when filtered by the less
      restrictive label \ensuremath{\ifc{l}}, thus the messages received (or dropped) from the
      queue are equivalent.
      \item[Case \textsc{I-new}] The newly allocated address can be the same
      for both \ensuremath{\ifc{\Varid{t}}_{1}} and \ensuremath{\ifc{\Varid{t}}_{2}}, thus resulting in \ensuremath{\ifc{l}}-equivalent states.
      \item[Case \textsc{I-write}] By \ensuremath{\varepsilon_{\ifc{l}}(\ifc{\Varid{t}}_{1})\mathrel{=}\Varid{earse}\;\ifc{l}\;\ifc{\Varid{t}}_{2}} both tasks
      write the same value, and therefore the resulting states will still
      be \ensuremath{\ifc{l}}-equivalent.
    \end{description}
    After \ensuremath{\ifc{\Varid{t}}_{2}} has taken a step, we finally arrive in the desired
    configuration \ensuremath{\ifc{c}_{2}'\mathrel{=}\ifc{\Sigma}_{2}'';\ifc{\Varid{ts}}_{2}'',\ifc{\Varid{ts}}_{2}''',\ifc{\Varid{ts}}_{2}''''}, where
    \ensuremath{\ifc{\Varid{ts}}_{2}''''} contains the task resulting from executing \ensuremath{\ifc{\Varid{t}}_{2}} (and
    might contain, zero (if the task was done or stuck), one (for most steps) or two tasks if a new task was launched).
    As required, we have
    \[ \ensuremath{\ifc{c}_{2}} \ensuremath{\hookrightarrow}^* \ensuremath{\ifc{c}_{2}''} \ensuremath{\hookrightarrow} \ensuremath{\ifc{c}_{2}'}
    \quad \land \quad \ensuremath{\ifc{c}_{1}'} \approx_{\ensuremath{\ifc{l}}} \ensuremath{\ifc{c}_{2}'} \text{.}\]
  \end{itemize}
  \qed
\end{proof}
With this, it is easy to proof Theorem~\ref{thm:restricted} as follows.
\begin{proof}[Proof of Theorem~\ref{thm:restricted}, TSNI]
  We proof the theorem by induction on the length of the derivation sequence in~\eqref{eq:tsni-lhs}.
  The base case for derivations
  of length 0 is trivial, allowing
  us to simple chose $\ensuremath{\ifc{c}_{2}'\mathrel{=}\ifc{c}_{2}}$.  In the step case, we assume
  the theorem holds for derivation sequences of length up to $n$, and show that it also
  holds for those of length $n+1$.  We split the derivation sequence from~\eqref{eq:tsni-lhs} as follows:
  \[
  \ensuremath{\ifc{c}_{1}} \ensuremath{\hookrightarrow} \ensuremath{\ifc{c}_{1}''} \ensuremath{\hookrightarrow}^n \ensuremath{\ifc{c}_{1}'}
  \]
  for some configuration \ensuremath{\ifc{c}_{1}''}.  By Lemma~\ref{lemma:rr-tsni-general}, we get
  \ensuremath{\ifc{c}''} with
  \begin{equation} \label{eq:tsni-proof-1}
  \ensuremath{\ifc{c}_{1}''} \approx_{\ensuremath{\ifc{l}}} \ensuremath{\ifc{c}_{2}''}
  \qquad \text{and} \qquad
  \ensuremath{\ifc{c}_{2}} \ensuremath{\hookrightarrow}^* \ensuremath{\ifc{c}_{2}''}
  \end{equation}
  Applying the induction hypothesis to
  $\ensuremath{\ifc{c}_{1}''} \ensuremath{\hookrightarrow}^n \ensuremath{\ifc{c}_{1}'}$, we get \ensuremath{\ifc{c}_{2}'} with
  \begin{equation} \label{eq:tsni-proof-2}
  \ensuremath{\ifc{c}_{1}'} \approx_{\ensuremath{\ifc{l}}} \ensuremath{\ifc{c}_{2}'}
  \qquad \text{and} \qquad
  \ensuremath{\ifc{c}_{2}''} \ensuremath{\hookrightarrow}^* \ensuremath{\ifc{c}_{2}'}
  \end{equation}
  Stitching together the derivation sequences from~\eqref{eq:tsni-proof-1} and~\eqref{eq:tsni-proof-2} directly gives
  us the right-hand side of the implication in the TSNI
  definition~\eqref{eq:tsni-rhs}, which concludes the proof.
  \qed
\end{proof}


\fi

\end{document}
